\documentclass[11pt,a4paper]{article}

\usepackage{hyperref}
\usepackage{cite}

\usepackage{amsmath}
\usepackage{amsfonts,amsthm,amssymb}

\numberwithin{equation}{section}

\newcommand{\beqa}{\begin{eqnarray}}
\newcommand{\eeqa}{\end{eqnarray}}

\newtheorem{theorem}{Theorem}[section]
\newtheorem{proposition}{Proposition}[section]

\newtheorem{corollary}{Corollary}[section]

{\theoremstyle{remark}
\newtheorem{rem}{Remark}[section]}
\newtheorem{identity}{Identity}
\newcommand{\tr}{\operatorname{tr}}

\newcommand{\End}{\operatorname{End}}

\newcommand{\bra}[1]{\langle\,#1\,|}
\newcommand{\ket}[1]{|\,#1\,\rangle}

\newcommand{\moy}[1]{\langle\,#1\,\rangle}

\def\sul{\sum\limits}
\def\pl{\prod\limits}

\newcommand{\la}{\lambda}

\begin{document}

\begin{flushright}
LPENSL-TH-06/16
\end{flushright}

\bigskip

\bigskip

\begin{center}
\textbf{{\Large The open XXX  spin chain in the SoV framework:}}

\textbf{{\Large scalar product of separate states}}

\vspace{45pt}

\begin{large}

{\bf N.~Kitanine}\footnote[1]{IMB UMR5584, CNRS, Univ. Bourgogne Franche-Comt\'e, F-21000 Dijon, France; Nikolai.Kitanine@u-bourgogne.fr},~~
{\bf J.~M.~Maillet}\footnote[2]{Univ Lyon, Ens de Lyon, Univ Claude Bernard, CNRS, 
Laboratoire de Physique, F-69342 Lyon, France;
 maillet@ens-lyon.fr},~~
{\bf G. Niccoli}\footnote[3]{Univ Lyon, Ens de Lyon, Univ Claude Bernard, CNRS, 
Laboratoire de Physique, F-69342 Lyon, France; giuliano.niccoli@ens-lyon.fr},

{\bf V.~Terras}\footnote[4]{LPTMS, CNRS, Univ. Paris-Sud, Universit\'e Paris-Saclay, 91405 Orsay, France;\\ veronique.terras@lptms.u-psud.fr}
\end{large}

\vspace{45pt}

\today

\end{center}

\vspace{45pt}

\begin{abstract}
We consider the XXX open spin-1/2 chain with the most general non-diagonal boundary terms, that we solve by means of the quantum separation of variables (SoV) approach. We compute the scalar products of separate states, a class of states which notably contains all the eigenstates of the model. As usual for models solved by SoV, these scalar products can be expressed as some determinants with a non-trivial dependance in terms of the inhomogeneity parameters that have to be introduced for the method to be applicable. We show that these determinants can be transformed into alternative ones in which the homogeneous limit can easily be taken. These new representations can be considered as generalizations of the well-known determinant representation for the scalar products of the Bethe states of the periodic chain. In the particular case where a constraint is applied on the boundary parameters, such that the transfer matrix spectrum and eigenstates can be characterized in terms of polynomial solutions of a usual $T$-$Q$ equation, the scalar product that we compute here corresponds to the scalar product between two {\it off-shell} Bethe-type states. If in addition one of the states is an eigenstate, 
the determinant representation can be simplified, hence leading in this boundary case to direct analogues of algebraic Bethe ansatz determinant representations of the scalar products for the periodic chain.  
\end{abstract}

\newpage

\section{Introduction}

The integrable quantum spin chains with boundary fields have been  attracting an increasing attention for several years \cite{AlcBBBQ87,Skl88,GhoZ94,JimKKKM95,Nep02,Nep04,CaoLSW03,YanZ07,Bas06,BasK07,KitKMNST07,KitKMNST08,CraRS10,CraRS11,FraSW08,FraGSW11,CaoYSW13b,Nic13,FalKN14,KitMN14,BelC13,Bel15,BelP15,AvaBGP15,BelP16}. In the most general case, they have long remained a paradoxal example of quantum integrable models for which no exact solution was known, and therefore were constituting a very challenging problem from the mere point of view of quantum integrability. Moreover, these systems happen to have interesting physical applications. They can notably be used to model out-of-equilibrium and transport properties in the spin chains \cite{Pro11}, hence leading to numerous applications in condensed matter physics. They are also related to widely studied classical stochastic models such as ASEP \cite{deGE05}. 

The algebraic framework for these open spin chains was formulated by Sklyanin \cite{Skl88} as an extension, based on the reflection equation formerly  introduced by Cherednik \cite{Che84}, of the  quantum inverse scattering method (QISM)\cite{FadST79}.
In this framework, the boundary fields are encoded into two boundary matrices which correspond to scalar solutions of the reflection equation  \cite{deVG93,GhoZ94}.
In the particular case in which both boundary matrices are diagonal, the spin chain Hamiltonian can be diagonalized by means of coordinate Bethe ansatz \cite{AlcBBBQ87}, or by means of (a boundary version of) algebraic Bethe ansatz (ABA) \cite{Skl88}.

The case in which the boundary matrices are non-diagonal turned out to be much more involved. Whereas the first stage of QISM, {\it i.e.}  the algebraic setting leading to the identification of the commuting conserved charges, does not differ for diagonal and non-diagonal  cases, the effective construction of the eigenstates is more difficult, and even 
seems to be not possible in the framework of the usual ABA approach 
for the completely generic non-diagonal case (there exist however some degenerate cases where ABA is still efficient \cite{BelCR13}).  There are in fact two essential difficulties for the ABA approach: there is no evident reference state and the transfer matrix mixes all four monodromy matrix entries.

The first successful description of the spectrum of a spin chain with non-diagonal boundary terms \cite{Nep02} was obtained using the fusion procedure \cite{KulRS81}. It was shown that, provided the left and right boundary parameters are related by some particular constraint, the transfer matrix spectrum  in the roots of unity points can be characterized in terms of polynomial solutions of a $T$-$Q$ equation of Baxter's type \cite{Bax82L}. Later \cite{Nep04}, the root of unity requirement turned out to be unnecessary. However, the boundary constraint remained essential for the description of the spectrum. The same constraint appeared naturally in the first ABA-like construction of (some of) the eigenstates \cite{CaoLSW03}. The latter were obtained using a local gauge transformation, following a procedure similar to the ABA solution of the XYZ spin chain \cite{Bax73a,FadT79}. This approach was later completed by the identification of a second reference state  \cite{YanZ07}. Several other methods (coordinate Bethe ansatz with elements of matrix product ansatz \cite{CraRS10,CraRS11}, 
$q$-Onsager algebra \cite{Bas06,BasK07} etc.) led to the same constraint as a necessary condition to characterize the spectrum  in terms of (polynomial) solutions of a $T$-$Q$ equation.
In  \cite{CaoYSW13b} it was suggested to circumvent this problem by allowing the $T$-$Q$ equation to admit some inhomogeneous term: hence, the spectrum in the unconstrained case can {\em a priori} still be described by polynomial solutions of such an equation provided the inhomogeneous term is adequately chosen.
A construction of the eigenstates, leading to the same inhomogeneous $T$-$Q$ equation, was also proposed in the framework of the so-called modified algebraic Bethe ansatz \cite{BelC13,Bel15,BelP15,AvaBGP15}. 

Open spin chains with non-diagonal boundary conditions have also been studied by means of the quantum Separation of Variables (SoV) \cite{FraSW08,FraGSW11,Nic12,Nic13,FalKN14,FalN14}.
This method, first introduced by Sklyanin in the QISM framework as an alternative to ABA for solving models in which a reference state cannot be identified \cite{Skl85}, has recently shown to be applicable to a large class of models \cite{NicT10,Nic10a,Nic11,GroN12,Nic13a,NicT15,LevNT16,NicT16}. It has in particular permitted to construct the complete set of eigenstates for the spin chains with the most general (unconstrained) boundary terms \cite{Nic12,FalKN14,FalN14}. It should be mentioned at this point that, within the SoV approach, the completeness of the eigenstates construction is intrinsic to the method:  this has to be compared to ABA in which the completeness of Bethe states is usually very difficult to prove (as notably for open spin chains with non-diagonal boundary terms). One of the inconvenience of the SoV approach, however, is that it applies to completely inhomogeneous versions of the model, and that the characterization of the transfer matrix spectrum and eigenstates {\em a priori} depends on the inhomogeneity parameters in a way which makes the homogeneous limit not so easy to recover. A reformulation in terms of solutions of a functional $T$-$Q$ equation has therefore to be worked out (see for instance \cite{NicT15,LevNT16}).  In the case of the spin chains with the most general (unconstrained) boundary terms, such a reformulation is presently not known, at least in terms of a usual $T$-$Q$ equation of Baxter's type. It is nevertheless possible, as shown in \cite{KitMN14}, to alternatively characterize the SoV spectrum and eigenstates of the model in terms of the solutions of the aforementioned inhomogeneous $T$-$Q$ equation, which also proves the completeness of this description. 


\medskip
In this paper, we pursue the study of the XXX open spin chain with non-diagonal boundary conditions in the SoV framework. Our aim is here to compute the scalar products of the so-called {\it separate states}, a class of states which notably includes all the eigenstates of the transfer matrix. This should be the first step towards the investigation of interesting physical quantities such as form factors and correlation functions as shown recently in a SoV settings in \cite{GroMN12} and \cite{Nic12,Nic13}. It is useful at this point to recall that, for models solved by ABA for which such a formula is known (such as for instance the closed spin chain with periodic boundary conditions), this determinant representation for the scalar products of Bethe states (or more precisely for the scalar product of one on-shell and one off-shell Bethe states) \cite{Sla89} has played a crucial role in the computation of form factors and correlation functions \cite{KitMT99,KitMT00}. In particular, the determinant representations for the form factors \cite{KitMT99} that follow using the resolution of the quantum inverse scattering problem \cite{KitMT99,MaiT00} and the determinant formula of the scalar products have been essential for the asymptotic analysis of the correlation functions \cite{KitKMST11b,KozT11,KitKMST12,KitKMT14,DugGK13} and for the numerical study of the dynamical structure factors \cite{CauHM05,CauM05,PerSCHMWA06} which leads to direct applications in condensed matter physics. 

In the case of open spin chains with diagonal boundaries, a determinant representation similar to \cite{Sla89} has also been obtained for the scalar products of Bethe states \cite{Wan02}, and has been used for the computation of the correlation functions  \cite{KitKMNST07,KitKMNST08}.  Unfortunately, the attempts to generalize this result to more general non-diagonal boundaries have so far remained unsuccessful, even in the case with a constraint which is {\em a priori} still solvable by Bethe ansatz:
the main problem here comes from the fact that the dual Bethe states could not be constructed in the usual ABA form.
Hence, the determinant representations that have been obtained in this context  \cite{FilK11,YanCFHHSZ11} are valid only for very special states in an over-constrained case (two boundary constraints instead of one). There are nevertheless strong reasons to expect that such a determinant representation for the scalar products should exist even for the most general boundary terms. Let us mention in particular \cite{DuvP15}, in which a determinant representation was obtained for the scalar products of two Bethe states in the case of a semi-infinite chain, and \cite{BelP16} in which such a type of formula, based on the construction of Bethe states through modified algebraic Bethe ansatz and on the description of the spectrum through an inhomogeneous $T$-$Q$ equation, was conjectured for the XXX chain with generic boundaries from the study of particular lattices with one or two sites.

As shown in the present paper, this problem can actually be solved for the open XXX chain by considering the model in the SoV framework. In fact, another considerable advantage of this method (apart from the established completeness of the eigenstate characterization) is that it leads by construction to determinant representations for the scalar product of two separate states, which can be seen in this context as the SoV counterpart of the scalar product of two (off-shell) Bethe states. However, these representations are, in their initial form, quite different from the aforementioned ABA type formula. In particular, they once again strongly depend on the (unphysical) inhomogeneity parameters that had to be introduced in the model to enable its solution by SoV, in a way that makes the consideration of the homogeneous (physical) limit not so obvious. They therefore need to be transformed into formulas in which the homogeneous limit can be taken explicitly. Ideally, these new representations should also, as the ABA type formula, be well adapted to the consideration of the thermodynamic limit: their rewriting should for instance involve in a natural way the roots of the $Q$-solution to the functional $T$-$Q$ equation which is used (at least in the constrained case when it can be clearly identified) to characterize the eigenstates. Let us mention here that such a reformulation has recently been performed in the case of the anti-periodic XXX chain \cite{KitMNT16}, for which an equivalence between the initial SoV representation of the scalar products and a Slavnov-type formula \cite{Sla89} (or a generalized version of it  \cite{FodW12}) has been established.

In this paper we consider the scalar products for the XXX chain with the most general boundary terms
in the framework of the SoV approach. 
We show that the scalar products of two generic separate states can be written as the determinant of a matrix which has a natural homogeneous limit. The obtained representation appears as some generalization of \cite{Sla89,FodW12}. In the case with the constraint, the corresponding separate states can be identified with some off-shell Bethe states (the SoV construction proving in that case the completeness of the ABA description), hence leading to an off-shell generalization of \cite{Sla89,FodW12}
\footnote{Although it was unfortunately not clearly written in that way in our previous paper \cite{KitMNT16}, it is worth noticing that such a generalization of the scalar product formula for two {\em off-shell} separate states can also be straightforwardly deduced from the SoV study of the anti-periodic XXX chain and the algebraic identities used in  \cite{KitMNT16}.}. This representation simplifies when one of the corresponding states becomes on-shell, {\it i.e.} when one of the separate states is associated with a polynomial $Q$-solution of the associated $T$-$Q$ equation.

The paper is organised  as follows.
In Section~\ref{sec-XXXopen} we introduce the open XXX spin-1/2 chain with the most general integrable boundary terms in the framework of the representation theory of the reflexion algebra.
In Section~\ref{sec-SOV} we  present the SoV construction of the eigenstates in the non-diagonal case, and explain how to reformulate the complete description of the spectrum and eigenstates that we obtain in this framework in terms of the $Q$-solutions of some inhomogeneous (in the unconstrained case) or homogeneous (in the constrained case) functional $T$-$Q$ equation. We introduce also  the notion of separate states which is central to this paper, and make the connexion with off-shell Bethe states.
In Section~\ref{sec-sp} we compute the scalar products of separate states. We introduce several useful identities which permit us to rewrite the initial SoV determinant representations in a much more convenient form (of Izergin type or Slavnov type) for the consideration of the homogeneous and thermodynamic limits.
Finally, we give in the first appendix a detailed description of the SoV basis, and in the second appendix a proof of the main identity enabling us to reformulate the scalar product of two separate states as a generalized version of \cite{Sla89,FodW12}.

\section{The open XXX spin chain}
\label{sec-XXXopen}

The open spin-1/2 XXX quantum spin chain with the most general non-diagonal
integrable boundary terms has the following Hamiltonian: 
\begin{multline} \label{H-XXX-Non-D}
H =\sum_{i=1}^{N-1}\Big[\sigma _{i}^{x}\sigma _{i+1}^{x}+\sigma _{i}^{y}\sigma_{i+1}^{y}+\sigma _{i}^{z}\sigma _{i+1}^{z}\Big]
+\frac{\eta }{\zeta _{-}}\Big[\sigma_{1}^{z}+2\kappa _{-}\left(e^{\tau _{-}}\sigma _{1}^{+}+e^{-\tau _{-}}\sigma
_{1}^{-}\right)\Big]  \\
  +\frac{\eta}{\zeta _{+}} \Big[\sigma _{N}^{z}+2\kappa _{+}\left(e^{\tau_{+}}\sigma _{N}^{+}+e^{-\tau _{+}}\sigma _{N}^{-}\right)\Big].
\end{multline}
Here $\sigma _{i}^{\alpha}$, $\alpha=x,y,z$, are local spin-$1/2$ operators (Pauli matrices) acting on the local quantum space $\mathcal{H}_i\simeq\mathbb{C}^2$ at site $i$, $\eta$ is a fixed arbitrary parameter, and the six complex boundary parameters $\zeta _{\pm }$, $\kappa _{\pm }$ and $\tau _{\pm }$ parametrize the coupling of the spin operators at site $1$ and $N$ with two arbitrary boundary magnetic fields.

The study of the Hamiltonian \eqref{H-XXX-Non-D} can be performed in the framework of the representation theory of the reflection algebra \cite{Skl88}.
The latter is an associative algebra defined by the generators $\mathcal{U}_{ij}(\lambda)$, $i,j=1,\ldots,n$, considered as the elements of a $n\times n$ square matrix $\mathcal{U}(\lambda)$, and by the relations 
\begin{equation}
R_{12}(\lambda-\mu)\,\mathcal{U}_{1}(\lambda)\,R_{12}(\lambda+\mu)\,\mathcal{U}_{2}(\mu)=\mathcal{U}_{2}(\mu)\,R_{12}(\lambda+\mu)\,\mathcal{U}_{1}(\lambda)\,R_{12}(\lambda-\mu),  \label{refl-eq}
\end{equation}
where $R(\lambda)\in\End(\mathbb{C}^n\otimes\mathbb{C}^n)$ is an $R$-matrix solution of the Yang-Baxter equation.
In \eqref{refl-eq}, the indices label as usual the space(s) of the tensor product $\mathbb{C}^n\otimes\mathbb{C}^n$ on which the corresponding matrix acts, {\it i.e.} for instance $\mathcal{U}_1(\lambda)=\mathcal{U}\otimes \mathrm{Id}$. The relation \eqref{refl-eq} is called reflection equation \cite{Che84}, or boundary Yang-Baxter equation.

In the present case, the $R$-matrix of the model,
\begin{equation}
R(\lambda )=\left( 
\begin{array}{cccc}
\lambda +\eta & 0 & 0 & 0 \\ 
0 & \lambda & \eta & 0 \\ 
0 & \eta & \lambda & 0 \\ 
0 & 0 & 0 & \lambda +\eta
\end{array}
\right)\
 \in\End(\mathbb{C}^2\otimes\mathbb{C}^2),  \label{Rmatrix}
\end{equation}
is the 6-vertex polynomial solution of the Yang-Baxter equation. The most general non-diagonal scalar solution $K(\lambda)\in\End(\mathbb{C}^2)$ of the reflection equation \eqref{refl-eq} associated with the $R$-matrix \eqref{Rmatrix} is \cite{deVG93,GhoZ94}
\begin{equation}
K(\lambda;\zeta ,\kappa ,\tau )=\frac{1}{\zeta }\left( 
\begin{array}{cc}
\zeta +\lambda & 2\kappa e^{\tau }\lambda \\ 
2\kappa e^{-\tau }\lambda & \zeta -\lambda
\end{array}
\right),  \label{Kxxx}
\end{equation}
which depends, in addition to the spectral parameter $\lambda$, on 3 different parameters $\zeta$, $\kappa$ and $\tau$.
Using this scalar solution one  can construct two different classes of
solutions to the reflection equation \eqref{refl-eq} in the 2$^{N}$-dimensional representation space
$\mathcal{H}=\otimes _{n=1}^{N}\mathcal{H}_{n}$ which corresponds to the physical space of states of the Hamiltonian \eqref{H-XXX-Non-D}.
More precisely, defining the two boundary matrices
\begin{equation}\label{Kpm}
K_{-}(\lambda )=K(\lambda-\eta/2 ;\zeta _{-},\kappa _{-},\tau _{-}),\qquad 
K_{+}(\lambda )=K(\lambda +\eta/2 ;\zeta _{+},\kappa _{+},\tau _{+}),
\end{equation}
where $\zeta _{\pm },\kappa _{\pm },\tau _{\pm }$ are the boundary
parameters appearing in \eqref{H-XXX-Non-D}, one can introduce the following boundary monodromy matrices:
\begin{align}
 &\mathcal{U}_{-}(\lambda ) =M_{0}(\lambda )\, K_{-}(\lambda )\, \hat{M}_{0}(\lambda )
=\left( 
\begin{array}{cc}
\mathcal{A}_{-}(\lambda ) & \mathcal{B}_{-}(\lambda ) \\ 
\mathcal{C}_{-}(\lambda ) & \mathcal{D}_{-}(\lambda )
\end{array}
\right) \in \text{End}(\mathcal{H}_{0}\otimes \mathcal{H}), 
  \label{U-}\\
&\mathcal{U}_{+}^{t_{0}}(\lambda ) 
=M_{0}^{t_{0}}(\lambda)\, K_{+}^{t_{0}}(\lambda )\, \hat{M}_{0}^{t_{0}}(\lambda )
=\left( 
\begin{array}{cc}
\mathcal{A}_{+}(\lambda ) & \mathcal{C}_{+}(\lambda ) \\ 
\mathcal{B}_{+}(\lambda ) & \mathcal{D}_{+}(\lambda )
\end{array}
\right) \in \text{End}(\mathcal{H}_{0}\otimes \mathcal{H}).
\label{U+}
\end{align}
Here $\mathcal{H}_0=\mathbb{C}^2$ is called auxiliary space, and $X^{t_0}$ denotes the transpose of the matrix $X$ in $\mathcal{H}_0$.
$M_{0}(\lambda )\in \End(\mathcal{H}_{0}\otimes \mathcal{H})$ stands for the
bulk monodromy matrix which, for an inhomogeneous chain of length $N$ with inhomogeneity parameters $\xi_1,\ldots,\xi_N\in\mathbb{C}$, is defined as the following ordered product of $R$-matrices,
\begin{equation}\label{mon}
M_{0}(\lambda )=R_{0N}(\lambda -\xi _{N}-\eta /2)\dots
R_{01}(\lambda -\xi _{1}-\eta /2)=\left( 
\begin{array}{cc}
A(\lambda ) & B(\lambda ) \\ 
C(\lambda ) & D(\lambda )
\end{array}
\right),
\end{equation}
and which satisfies the quadratic relation:
\begin{equation}\label{RTT}
R_{12}(\lambda -\mu )\, M_{1}(\lambda )\, M_{2}(\mu )=M_{2}(\mu )\, M_{1}(\lambda)\, R_{12}(\lambda -\mu ).
\end{equation}
In \eqref{U-}-\eqref{U+} we have also used  the notation:
\begin{equation}
\hat{M}_0(\lambda )=(-1)^{N}\,\sigma _{0}^{y}\,M_0^{t_0}(-\lambda)\,\sigma _{0}^{y}.
\end{equation}
Then, the two matrices $\mathcal{V}_-(\lambda)= \mathcal{U}_{-}(\lambda+\eta/2 )$ and $\mathcal{V}_{+}(\lambda )=\mathcal{U}_{+}^{t_{0}}(-\lambda-\eta/2 )$ are two solutions of the
reflection equation \eqref{refl-eq}.
They enable one to define a one-parameter family of transfer matrices,
\begin{align}
\mathcal{T}(\lambda )
 &=\tr_{0}\{K_{+}(\lambda )\,M(\lambda)\,K_{-}(\lambda )\, \hat{M}(\lambda )\}
   \nonumber\\
 &=\tr_{0}\left\{ K_{+}(\lambda)\, \mathcal{U}_{-}(\lambda )\right\} 
   =\tr_{0}\left\{ K_{-}(\lambda )\, \mathcal{U}_{+}(\lambda )\right\} \in \text{End}(\mathcal{H}),
\label{transfer}
\end{align}
which are commuting operators on $\mathcal{H}$.
The Hamiltonian \eqref{H-XXX-Non-D} of the XXX spin-1/2 Heisenberg chain with boundary conditions given by $\zeta _{\pm }$, $\kappa _{\pm }$ and $\tau _{\pm }$ can then be obtained in the
homogeneous limit ($\xi _{m}=0$ for $m=1,\ldots ,N$) as the
following derivative of the transfer matrix \eqref{transfer}:
\begin{equation}
H=\frac{2\,\eta ^{1-2N}}{\tr\{K_{+}(\eta /2)\}\,\tr\{K_{-}(\eta/2)\}}\,
\frac{d}{d\lambda }\mathcal{T}(\lambda )_{\,\vrule height13ptdepth1pt\>%
{\lambda =\eta /2}\!}+\text{constant.}  \label{Htxxx}
\end{equation}

We finally recall the inversion relation for the boundary monodromy matrix $\mathcal{U}_-(\lambda )$ \eqref{U-},
\begin{equation}\label{invert-U-}
   \mathcal{U}_-(\lambda +\eta/2)\ \mathcal{U}_-(-\lambda +\eta/2)
   = \frac{\mathrm{det}_{q}\,\mathcal{U}_-(\lambda )}{2\lambda-2\eta},
\end{equation}
where $\mathrm{det}_{q}\,\mathcal{U}_-(\lambda )$ is the quantum determinant, which is a central element of the corresponding boundary algebra: $\big[ \mathrm{det}_{q}\,\mathcal{U}_-(\lambda ), \mathcal{U}_-(\mu )\big]=0$.
It can be expressed  as
\begin{equation}\label{detqU-}
   \mathrm{det}_{q}\,\mathcal{U}_-(\lambda )=\mathrm{det}_{q} M(\lambda)\, \mathrm{det}_{q}M(-\lambda)\,\mathrm{det}_{q}K_-(\lambda).
\end{equation}
Here
\begin{equation}\label{detqM}
   \mathrm{det}_q M(\lambda)= a(\lambda+\eta/2)\, d(\lambda-\eta/2),
\end{equation}
denotes the bulk quantum determinant, with
\begin{equation}\label{a-d}
a(\lambda )\equiv \prod_{n=1}^{N}(\lambda -\xi _{n}+\eta /2),
\qquad 
d(\lambda )\equiv \prod_{n=1}^{N}(\lambda -\xi _{n}-\eta /2),
\end{equation}
whereas $\mathrm{det}_{q}K_-(\lambda)$ denote the quantum determinant of the scalar boundary matrix $K_-(\lambda)$. Similar relations hold for the boundary monodromy matrix $\mathcal{U}_+(\lambda )$ \eqref{U+} and for its quantum determinant $\mathrm{det}_{q}\,\mathcal{U}_+(\lambda )$, which can easily be deduced from the previous ones using the fact that 
$\mathcal{U}_{+}^{t_{0}}(-\lambda )$ satisfies the same algebra as $\mathcal{U}_-(\lambda )$.
The quantum determinants of the scalar boundary matrices $K_\mp(\lambda)$ \eqref{Kpm} are explicitly given as
\begin{equation}\label{detqK-}
    \mathrm{det}_{q}K_\mp(\lambda)
    = \mp 2(\lambda\mp\eta)\left(\frac{1+4\kappa_\mp^2}{\zeta_\mp^2}\,\lambda^2-1\right).
\end{equation}
%
%
%

\section{Spectrum and eigenstates by Separation of Variables}
\label{sec-SOV}

For generic values of the inhomogeneity parameters $\xi_n$, $1\le n\le N$, and of the boundary parameters $\zeta_\pm$, $\kappa_\pm$ and $\tau_\pm$, the spectrum and eigenstates of the transfer matrix \eqref{transfer} can be completely characterized in the framework of the quantum version of the Separation of Variables approach \cite{Skl85,Skl90,Skl92}. This method can be directly applied if one of the boundary matrices is triangular  \cite{Nic12}. However due to the $GL(2,\mathbb{C})$ symmetry of the $R$-matrix \eqref{Rmatrix} the most general solutions  \eqref{Kpm} of the reflection equation \eqref{refl-eq} can always be reduced to a case with (at least) one triangular boundary matrix (see also \cite{FraSW08}). 

\subsection{Reduction to a triangular case}

The $R$-matrix \eqref{Rmatrix} of the XXX spin-1/2 chain satisfies, for any invertible $2\times 2$ scalar matrix $W$,  the
following symmetry property (scalar Yang-Baxter equation):
\begin{equation}
R_{12}(\lambda )\,W_{1}\,W_{2}=W_{2}\,W_{1}\,R_{12}(\lambda ).  \label{GL2-Sym}
\end{equation}
%
For any such matrix $W\in GL(2,\mathbb{C})$, we can apply the corresponding global gauge  transformation to the boundary monodromy matrices:
\begin{equation}
   \mathcal{\bar{U}}_{\mp }(\lambda)
   =W_{0}\, \Gamma_W\, \mathcal{U}_{\mp }(\lambda )\,\Gamma_W^{-1}\, W_{0}^{-1}
   =\left( 
\begin{array}{cc}
\mathcal{\bar{A}}_{\mp}(\lambda ) & \mathcal{\bar{B}}_{\mp}(\lambda ) \\ 
\mathcal{\bar{C}}_{\mp}(\lambda ) & \mathcal{\bar{D}}_{\mp}(\lambda )
\end{array}
\right) ,
\label{GU1}
\end{equation}
where $W_0$ acts on the auxiliary space, whereas $\Gamma_W \equiv \otimes_{n=1}^N W_n$ acts on the quantum space of states. Due to \eqref{GL2-Sym}, the matrices $\mathcal{\bar{U}}_{\mp }(\lambda)$ satisfy the same reflexion equations as $\mathcal{U}_{\mp }(\lambda)$. They can be expressed in terms of the gauge transformed scalar boundary matrices
\begin{equation}
\bar{K}_{\mp }(\lambda )=W_{0} \,K_{\mp }(\lambda )\,W_{0}^{-1}
=\left( 
\begin{array}{cc}
\bar{a}_{\mp }(\lambda ) & \bar{b}_{\mp }(\lambda ) \\ 
\bar{c}_{\mp }(\lambda ) & \bar{d}_{\mp }(\lambda )
\end{array}
\right) ,
\label{XXX-Similarity}
\end{equation}
in the following form
\begin{align}
\mathcal{\bar{U}}_{-}(\lambda ) 
&=M(\lambda )\,\bar{K}_{-}(\lambda )\,\hat{M}(\lambda ) , 
\\
\mathcal{\bar{U}}_{+}(\lambda ) 
&=\left( M^{t_{0}}(\lambda )\,\bar{K}_{+}^{t_{0}}(\lambda )\,\hat{M}^{t_{0}}(\lambda )\right) ^{t_{0}} .
\end{align}
The corresponding gauge transformed  transfer matrix $\mathcal{\bar{T}}(\lambda )=\Gamma_W\, \mathcal{T}(\lambda )\,  \Gamma_W^{-1}$ is then given as
\begin{equation}\label{gauge-transfer}
\mathcal{\bar{T}}(\lambda )=\tr_{0}\left\{\bar{K}_{+}(\lambda )\,M(\lambda)\,\bar{K}_{-}(\lambda )\,\hat{M}(\lambda )\right\},
\end{equation}
or, in terms of the matrix elements of the gauged transformed boundary monodromy matrix $\mathcal{\bar{U}}_{-}(\lambda ) $, as
\begin{equation}
  \mathcal{\bar{T}}(\lambda )= \bar{a}_+(\lambda ) \, \mathcal{\bar{A}}_{-}(\lambda )
  +\bar{b}_+(\lambda ) \, \mathcal{\bar{C}}_{-}(\lambda ) + \bar{c}_+(\lambda ) \, \mathcal{\bar{B}}_{-}(\lambda )
  +\bar{d}_+(\lambda ) \, \mathcal{\bar{D}}_{-}(\lambda ).
\end{equation}

The $GL(2,\mathbb{C})$ symmetry can be used to transform the boundary matrices into triangular ones.
%
%
So as to obtain an algebraic framework directly solvable by SoV, we can for instance choose a gauge transformation such that the new boundary matrix $\bar{K}_+(\lambda)$ \eqref{XXX-Similarity} is lower triangular, {\it i.e.}   $\bar{b}_{+}(\lambda )=0$, the other entries of  $\bar{K}_-(\lambda)$ and  $\bar{K}_+(\lambda)$ being allowed to take any  form provided that $\bar{b}_{-}(\lambda )\neq 0$ \cite{Nic12}.
To this aim, we consider the following $2\times 2$ invertible matrix,
\begin{equation}\label{W}
W\equiv W_{\epsilon_+,\epsilon_-}
=\left( 
\begin{array}{cc}
 1 & -\frac{1-\epsilon_+\sqrt{1+4\kappa_+^2} }{2\kappa_+ e^{-\tau_{+}} }\\ 
 \frac{1-\epsilon_-\sqrt{1+4\kappa_-^2} }{2\kappa_- e^{\tau _{-}} } & 1
\end{array}
\right) ,
\end{equation}
%
%
%
for a given
choice of $(\epsilon_+,\epsilon_-)\in\{-1,1\}^2$.
Then the gauge transformed boundary matrices \eqref{XXX-Similarity} are of the form
\begin{equation}
\bar{K}_+(\lambda )=\mathrm{I}+\frac{\la+\eta/2}{\bar{\zeta}_+}(\sigma^z+\bar{\mathsf{c}}_{+}\sigma^-),\qquad \bar{K}_-(\lambda )=\mathrm{I}+\frac{\la-\eta/2}{\bar{\zeta}_-}(\sigma^z+\bar{\mathsf{b}}_{-}\sigma^+),
\label{triangKnobar}
\end{equation}
with
\begin{align}
   &\bar{\zeta}_{\pm}=  \epsilon_\pm\, \frac{\zeta_\pm  }{ \sqrt{1+4\kappa_\pm^2}  }\, ,
\label{zeta-bar}\\
  &\bar{\mathsf{c}}_{+}
  = \epsilon_+\, \frac{ 2\kappa_+ e^{-\tau_+} }{ \sqrt{1+4\kappa_+^2} }
     \left[ 1
     +\frac{ \big(1+\epsilon_+\sqrt{1+4\kappa_+^2} \big)\big(1-\epsilon_-\sqrt{1+4\kappa_-^2} \big)}
               {4\kappa_+\kappa _{-}e^{\tau _- -\tau_+}  } \right],
\\
 & \bar{\mathsf{b}}_{-}
  = \epsilon_-\, \frac{ 2\kappa_- e^{\tau_-} }{ \sqrt{1+4\kappa_-^2} }
     \left[ 1
     +\frac{ \big(1-\epsilon_+\sqrt{1+4\kappa_+^2} \big)\big(1+\epsilon_-\sqrt{1+4\kappa_-^2} \big)}
               {4\kappa_+\kappa _{-}e^{\tau _- -\tau_+}  } \right]. 
                \label{Tri-Tri-gauge}
\end{align}
Note that the  quantum determinants \eqref{detqK-} can be expressed in terms of the new parameters \eqref{zeta-bar} as
\begin{equation}\label{detqKtri}
    \mathrm{det}_{q}K_\pm(\lambda)=  \mathrm{det}_{q}\bar{K}_\pm(\lambda)
    = \pm 2(\lambda\pm\eta)\left(\frac {\lambda^2}{\bar{\zeta}_\pm^2}\,-1\right).
\end{equation}

Due to the symmetry of the $R$-matrix \eqref{GL2-Sym}, the transfer matrix of the original problem $\mathcal{T}(\lambda )$ has the same spectrum  as the gauge transformed transfer matrix $\mathcal{\bar{T}}(\lambda )$ \eqref{gauge-transfer}. After the global gauge transformation the eigenstates can be constructed in terms of  the separate variables diagonalizing $\mathcal{\bar{B}}_{-}(\lambda )$ provided that $\bar{\mathsf{b}}_{-}\not=0$, or alternatively   by introducing the separate variables  diagonalizing $\mathcal{\bar{C}}_{+}(\lambda )$ provided that $\bar{\mathsf{c}}_+\not=0$. If both coefficients are zero there is no need to apply the separation of variables as the model is equivalent to a case with diagonal boundary terms which can be solved by algebraic Bethe ansatz. In all other cases, due to the free choice of the two sign parameters  $\epsilon_\pm$, we can always  choose a gauge transformation in such a way that $\bar{\mathsf{b}}_{-}\not=0$.

The eigenstates of the initial problem $\ket{\Psi}$ can be therefore expressed in terms of the eigenstate of the new triangular case $\ket{\bar{\Psi}}$ using the global gauge transformation $\Gamma_W$,
\begin{equation}
\ket{\Psi}=\Gamma_W^{-1}\,\ket{\bar{\Psi}},\qquad\bra{\Psi} =\bra{\bar{\Psi}}\,\Gamma_W.
\end{equation}
The aim of this paper is the computation of scalar products of the off-shell and on-shell states. It is evident that for such computations the global gauge transformation $\Gamma_W$ will be always canceled out.


\subsection{
Construction of the eigenstates in the SoV framework}

We shall suppose from now on that  $\bar{\mathsf{b}}_{-}\not=0$ and that the inhomogeneity parameters are generic, or more precisely that they satisfy the condition 
\begin{equation}
\xi_j,\xi _j \pm \xi _k\notin\{0,-\eta,\eta \},
\quad
\forall j,k\in \{1,\ldots,N\},\ j\neq k.
  \label{cond-inh}
\end{equation}
%
Under these hypotheses we can, as in \cite{Nic12}, construct a basis
\begin{equation}\label{right-basis}
   \big\{\, \ket{\mathbf{h}_-}, \ \mathbf{h}\equiv(h_1,\ldots,h_N)\in\{0,1\}^N\, \big\}
\end{equation}
of the space of states and a basis
\begin{equation}\label{left-basis}
   \big\{\, \bra{\mathbf{h}_-}, \ \mathbf{h}\equiv(h_1,\ldots,h_N)\in\{0,1\}^N\, \big\}
\end{equation}
of the dual space of states which diagonalize the operator family $\bar{\mathcal{{B}}}_{-}(\lambda )$ independently of the spectral parameter $\lambda$ (see Appendix~\ref{app-Beigen} for details).
The two bases \eqref{right-basis} and \eqref{left-basis} as constructed in Appendix~\ref{app-Beigen} are orthogonal with respect to the canonical scalar product in the spin basis:
\begin{equation}\label{norm}
     \moy{\mathbf{h}'_- \mid\mathbf{h}_-}
     = \delta_{\mathbf{h},\mathbf{h}'}\,\frac{  N_{\boldsymbol{\xi},-} }{\widehat{V}\big( \xi_1^{(h_1)} ,\ldots,\xi_N^{(h_N)}\big)}\,  .
\end{equation}
According to the normalization chosen in Appendix~\ref{app-Beigen}, the normalization factor $N_{\boldsymbol{\xi},-}$, which depends on the $N$-tuple $\boldsymbol{\xi}\equiv(\xi_1,\ldots,\xi_N)$ and on the parameters ${\bar{\zeta}}_-$, ${\bar{\mathsf{b}}}_{-}$ of the boundary matrix $\bar{K}_-$ \eqref{triangKnobar},  is
\begin{equation}\label{norm-factor}
   N_{\boldsymbol{\xi},-} =\widehat{V}( \xi_1,\ldots, \xi_N)\,
    \frac{\widehat{V}( \xi_1^{(0)},\ldots, \xi_N^{(0)})}{\widehat{V}( \xi_1^{(1)},\ldots, \xi_N^{(1)})}
    \,\prod_{n=1}^N\frac{\xi_n\,{\bar{\mathsf{b}}}_{-}}{\xi_n-{\bar{\zeta}}_-}.
\end{equation}
In \eqref{norm} and \eqref{norm-factor} we have used the following shorthand notations, that we will use throughout the whole paper:
\begin{equation}\label{xi-h}
    \xi_n^{(h)}=\xi_n+\eta/2-h\eta, \qquad 1\le n\le N, \quad h\in\{0,1\},
\end{equation}
for shifted inhomogeneity parameters, and
\begin{equation}\label{VDM}
\widehat{V}( x_1,\ldots, x_N)=\det_{1\le i,j\le N}\big[ x_i^{2(j-1)}\big]=\pl_{j<k}( x_k^2- x_j^2),
\end{equation}
for the Vandermonde determinant of a $N$-tuple of square variables $(x_1^2,\ldots, x_N^2)$.

It follows from the algebraic construction of the transfer matrix $\mathcal{T}(\lambda )$  \eqref{transfer} that it is a polynomial function  of degree $N+1$ in the variable $\lambda ^{2}$. Its leading coefficient is given by,
\begin{equation*}
  t_{N+1}\,\lambda^{2(N+1)}\, \mathrm{Id},
  \quad \text{with}
  \quad t_{N+1}= \frac{2}{\zeta_+\zeta_-}[1+4\kappa_+\kappa_-\cosh(\tau_+-\tau_-)]
                        =\frac{2+\bar{\mathsf{b}}_{-}\bar{\mathsf{c}}_{+}}{\bar{\zeta}_+\,\bar{\zeta}_-}.
\end{equation*}
By using the SoV basis \eqref{right-basis} and \eqref{left-basis}, it is possible to completely characterize its spectrum and eigenstates in terms of solutions of a set of discrete equations, as stated in the following theorem.

\begin{theorem}\label{th-Sp-0}
Let  the inhomogeneity parameters be generic   \eqref{cond-inh} and  ${\bar{\mathsf{b}}}_{-}\not=0$.

Then, the spectrum $\Sigma _{\mathcal{T}}$ of $\mathcal{T}(\lambda )$ is simple and
coincides with the set of functions of the form
\begin{multline}\label{form-t}
t(\lambda ) 
=t_{N+1} \left( \lambda ^{2}- (\eta/2)^2 \right)
\prod_{b=1}^{N}\big(\lambda^{2}-\xi_{b}^{2} \big)  
 +2 (-1)^N\mathrm{det}_q M(0)
 \prod_{b=1}^{N}\frac{\lambda ^{2}-\xi_{b}^{2}}{\big(\eta /2\big)^{2}-\xi_{b}^{2}}
 \\
 +\sum_{a=1}^{N}\frac{4\lambda ^{2}-\eta ^{2}}{4 \xi_{a}^{2}-\eta ^{2}}\,
 \prod_{\substack{ b=1  \\ b\neq a}}^{N}
 \frac{\lambda ^{2}-\xi_{b}^{2}}{\xi_{a}^{2}- \xi_{b}^{2}}\,
t(\xi _{a}),
\end{multline}
which satisfy the discrete system of equations
\begin{equation}
\det \left( 
\begin{array}{cc}
t\big(\xi _{n}^{(0)}\big) & -\mathsf{A}_{{\bar{\zeta}}_+,{\bar{\zeta}}_-}\big(\xi _{n}^{(0)}\big)
\\ 
-\mathsf{A}_{{\bar{\zeta}}_+,{\bar{\zeta}}_-}\big(-\xi _{n}^{(1)}\big) & t\big(\xi _{n}^{(1)}\big)
\end{array}
\right) =0,\quad  \forall n\in \{1,\ldots,N\},
\label{ARXFI-Functional-eq}
\end{equation}
in terms of the function\footnote{Note that this function is related to the quantum determinants by
\begin{equation*}
\frac{\det_{q}K_{\pm }(\lambda )\, \det_{q}\mathcal{U}_{\mp }(\lambda )}{\eta^2-4\lambda ^{2}} 
=\mathsf{A}_{\bar{\zeta}_{+},\bar{\zeta}_{-}}(\lambda +\eta /2)\, \mathsf{A}_{\bar{\zeta}_{+},\bar{\zeta}_{-}}(-\lambda
+\eta /2).
\end{equation*}
}
\begin{equation}\label{Aeps_pm}
\mathsf{A}_{\bar{\zeta}_{+},\bar{\zeta}_{-}}(\lambda ) 
\equiv (-1)^N  \frac{2\lambda+\eta }{2\lambda }\,
\frac{(\lambda-\frac{\eta}{2}+\bar{\zeta}_+)(\lambda-\frac{\eta}{2}+\bar{\zeta}_-)}{\bar{\zeta}_+\,\bar{\zeta}_-}\,
a(\lambda )\, d(-\lambda ).
\end{equation}
%

The one-dimensional right and left $\mathcal{T}(\lambda)$-eigenstates associated with the eigenvalue $t(\lambda )\in \Sigma _{\mathcal{T}}$ are respectively generated by the vectors
\begin{align}
  &\ket{\Psi_t}
  =\sum_{\mathbf{h}\in\{0,1\}^N}\prod_{n=1}^{N}Q_{t}(\xi _{n}^{(h_n)})\
   \widehat{V}\big( \xi_1^{(h_1)},\ldots,\xi_N^{(h_N)} \big) 
  \, \Gamma_W^{-1}  \,  \ket{\mathbf{h}_-} ,  
   \label{eigenT-right}\\
  &\bra{\Psi_t}
  =\sum_{\mathbf{h}\in\{0,1\}^N}
    \prod_{n=1}^{N}\left[\left(\frac{\xi_n-\eta}{\xi_n+\eta}\frac{\mathsf{A}_{\bar{\zeta}_{+},\bar{\zeta}_{-}}(\xi_n^{(0)} ) }{\mathsf{A}_{\bar{\zeta}_{+},\bar{\zeta}_{-}}(-\xi_n^{(1)} ) }\right)^{\! h_n} Q_{t}(\xi _{n}^{(h_n)}) \right]
    \nonumber\\
    &\hspace{7cm}\times
   \widehat{V}\big( \xi_1^{(h_1)},\ldots,\xi_N^{(h_N)} \big) 
   \,\bra{\mathbf{h}_-}
    \, \Gamma_W.  
   \label{eigenT-left}
\end{align}
In these expressions 
$\ket{\mathbf{h}_-}$ and $\bra{\mathbf{h}_-}$ denote the eigenvectors \eqref{right-SOV-state} and \eqref{left-SOV-state} of the $W$-gauge transformed 
 operator  $\bar{\mathcal{B}}_{-}( \lambda) $, see \eqref{GU1}, and $Q_t$ is a function on the discrete set of values $\xi_n^{(h_n)},\ n\in\{1,\ldots,N\},\ h_n\in\{0,1\}$, which satisfies
\begin{equation}
\frac{Q_{t}(\xi _n^{(1)})}{Q_{t}(\xi _n^{(0)})}
=\frac{t(\xi _n^{(0)})}{\mathsf{A}_{\bar{\zeta}_{+},\bar{\zeta}_{-}}(\xi _n^{(0)})}
=\frac{\mathsf{A}_{\bar{\zeta}_{+},\bar{\zeta}_{-}}(-\xi _n^{(1)})} {t(\xi _n^{(1)})},
\qquad
n=1,\ldots, N.
\label{Q-dis}
\end{equation}
\end{theorem}


For further study, it is convenient to decompose the ratio appearing in the expression of the left eigenstate \eqref{eigenT-left} as a product of two factors: a factor which explicitly depends on the boundary parameters $\bar{\zeta}_+$ and $\bar{\zeta}_-$ and a factor which does not depend on these parameters,
\begin{equation}
   \frac{\xi_n-\eta}{\xi_n+\eta}\frac{\mathsf{A}_{\bar{\zeta}_{+},\bar{\zeta}_{-}}(\xi_n^{(0)} ) }{\mathsf{A}_{\bar{\zeta}_{+},\bar{\zeta}_{-}}(-\xi_n^{(1)} ) }
   =f_n\, g_n\, ,
\end{equation}
with
\begin{equation}\label{g_n}
   g_n\equiv g_{\bar{\zeta}_+,\bar{\zeta}_-}(\xi_n)
   =\frac{(\xi_n+\bar{\zeta}_+)(\xi_n+\bar{\zeta}_-)}{(\xi_n-\bar{\zeta}_+)(\xi_n-\bar{\zeta}_-)},
\end{equation}
and
\begin{align}
 f_{n}\equiv f(\xi_n,\{\xi\})
 &= -\prod_{\substack{a=1 \\ a\neq n}}^{N}\frac{(\xi_n-\xi_a+\eta)(\xi_n+\xi_a+\eta)}{(\xi_n-\xi_a-\eta)(\xi_n+\xi_a-\eta)}
 \nonumber\\
 &= -\prod_{\substack{a=1 \\ a\neq n}}^{N}
\frac{\left[ \big(\xi_{n}^{(0)}\big)^2-\big(\xi _{a}^{(1)}\big)^2\right]
         \left[ \big(\xi _{n}^{(0)}\big)^2-\big(\xi_{a}^{(0)}\big)^2\right] }
        {\left[ \big(\xi _{n}^{(1)}\big)^2-\big(\xi _{a}^{(1)}\big)^2\right]
         \left[ \big(\xi _{n}^{(1)}\big)^2-\big(\xi _{a}^{(0)}\big)^2\right] }.
         \label{f_n}
\end{align}
It is also interesting to notice that the factors $f_n$ above can be combined with the Vandermonde determinant in the expression of the  left eigenstate \eqref{eigenT-left} by using
the following identity, which can  be proven by direct computation:
\begin{equation}\label{id-VDM}
    \prod_{n=1}^N(-f_n)^{h_n}\cdot \widehat{V}\big( \xi_1^{(h_1)},\ldots,\xi_N^{(h_N)} \big)
    =
    \frac{\widehat{V}\big( \xi_1^{(0)},\ldots,\xi_N^{(0)} \big)}{\widehat{V}\big( \xi_1^{(1)},\ldots,\xi_N^{(1)} \big)}    \,
    \widehat{V}\big( \xi_1^{(1-h_1)},\ldots,\xi_N^{(1-h_N)} \big).
\end{equation}
%
%
%
This identity enables one to rewrite the left eigenstate \eqref{eigenT-left} as
\begin{align}
\bra{\Psi_t}
  =  \frac{\widehat{V}( \xi_1^{(0)},\ldots, \xi_N^{(0)})}{\widehat{V}( \xi_1^{(1)},\ldots, \xi_N^{(1)})}
  \sum_{\mathbf{h}\in\{0,1\}^N} & \prod_{n=1}^{N}\left[(-g_n)^{ h_n} \, Q_{t}(\xi _{n}^{(h_n)}) \right]
    \notag\\
    \times
    &\widehat{V}\big( \xi_1^{(1-h_1)},\ldots,\xi_N^{(1-h_N)} \big)
   \,\bra{\mathbf{h}_-}
   \, \Gamma_W .  
   \label{eigenT-left-bis}
\end{align}

\subsection{Characterization in terms of solutions of a functional $T$-$Q$  equation}

As usual in the SoV framework, the transfer matrix spectrum and eigenstates are characterized in terms of functions $Q_t$, solutions of \eqref{Q-dis}, which are defined on the discrete set of shifted inhomogeneity parameters only.
So as to reformulate this characterization in a more convenient form for the homogeneous limit, we would like to extend the discrete equation \eqref{Q-dis} into a functional equation for a function that we still denote by $Q_t$ defined on the whole complex plane $\mathbb{C}$.

In the case of the open XXZ spin-1/2 chain with generic integrable boundary conditions, this problem was considered in \cite{KitMN14}. It was shown there that it is possible to reformulate the discrete SoV characterization of the transfer matrix spectrum in terms of particular polynomial solutions of a functional equation which happens to be an inhomogeneous version of Baxter's usual $T$-$Q$ equation. For the XXX case that we consider here, this result can be formulated as follows.

\begin{theorem}
\label{th-Sp-1}
Let the inhomogeneities $\xi _{1},\ldots,\xi _{N}$ be generic \eqref{cond-inh},
and let us suppose moreover that $\bar{\mathsf{c}}_{+}\not=0$.
%
%
Then $t(\lambda )\in \Sigma _{\mathcal{T}}$ if and only if there exists a unique function $Q_t(\lambda)$ of the form
\begin{equation}\label{Q-form}
   Q_{t}(\lambda )=\prod_{b=1}^{N}\left( \lambda ^{2}-\lambda_{b}^{2}\right) ,
   \qquad
   \lambda_1,\ldots,\lambda_N\in\mathbb{C}\setminus \big\{ \pm \xi_1^{(0)},\ldots,\pm \xi_N^{(0)} \big\},
\end{equation}
such that
\begin{equation}
t(\lambda )\, Q_{t}(\lambda )
=\mathsf{A}_{\bar{\zeta}_{+},\bar{\zeta}_{-}}(\lambda)\, Q_{t}(\lambda -\eta )
+\mathsf{A}_{\bar{\zeta}_{+},\bar{\zeta}_{-}}(-\lambda)\, Q_{t}(\lambda +\eta )
+F(\lambda ),
\label{Inhom-BAX-eq}
\end{equation}
with
\begin{equation}
F(\lambda )
=\frac{\bar{\mathsf{b}}_{-}\bar{\mathsf{c}}_{+}}{\bar{\zeta} _{-}\bar{\zeta} _{+}}
 \left( \lambda ^{2}-\left( \eta /2\right)^{2}\right) \,
\prod_{b=1}^{N}\prod_{h=0}^{1}\left( \lambda ^{2}-\big( \xi_{b}^{(h)}\big)^{2}\right) .  \label{DEF-F}
\end{equation}
Similarly, $t(\lambda )\in \Sigma _{\mathcal{T}}$ if and only if there exists a unique function $P_t(\lambda)$ of the form
\begin{equation}\label{P-form}
   P_{t}(\lambda )=\prod_{b=1}^{N}\left( \lambda^{2}-\mu_{b}^{2}\right) ,
   \qquad
   \mu_1,\ldots,\mu_N\in\mathbb{C}\setminus \big\{ \pm \xi_1^{(0)},\ldots,\pm \xi_N^{(0)} \big\},
\end{equation}
such that
\begin{equation}
t(\lambda )\, P_{t}(\lambda )
=\mathsf{A}_{-\bar{\zeta}_{+},-\bar{\zeta}_{-}}(\lambda)\, P_{t}(\lambda -\eta )
+\mathsf{A}_{-\bar{\zeta}_{+},-\bar{\zeta}_{-}}(-\lambda)\, P_{t}(\lambda +\eta )
+F(\lambda ).
\label{Inhom-BAX-eq-bis}
\end{equation}
\end{theorem}

\begin{proof}
The proof is similar to the one in \cite{KitMN14}. The fact that we obtain two different characterizations follows from the invariance of the two functions $t(\lambda)$ and $F(\lambda)$ through a change of $(\bar{\zeta}_+,\bar{\zeta}_-)$  into $(-\bar{\zeta}_+,-\bar{\zeta}_-)$, the function $\mathsf{A}_{\bar{\zeta}_{+},\bar{\zeta}_{-}}(\lambda)$ being transformed into $\mathsf{A}_{-\bar{\zeta}_{+},-\bar{\zeta}_{-}}(\lambda)$.
\end{proof}

If instead $\bar{\mathsf{c}}_{+}=0$ (which is equivalent to the boundary constraint\footnote{The boundary constraint which ensures the cancellation condition for the function \eqref{DEF-F} for some choice of $\epsilon_+,\epsilon_-$ can be rewriten in terms of the boundary coefficients of the initial problem as
\begin{equation} \label{cond-hom-eq}
  1+4\kappa _{+}\kappa _{-}\cosh (\tau _{+}-\tau _{-})=\pm \sqrt{(1+4\kappa _{+}^{2})(1+4\kappa _{-}^{2})} \, .
\end{equation} } \cite{Nep02,Nep04}) %
the function \eqref{DEF-F}  vanishes. As stated below, the transfer matrix spectrum can in that case be completely characterized by polynomial solutions of a usual ({\it i.e.} homogeneous) functional $T$-$Q$ equation. In fact, we also get two possible equivalent complete characterizations of this type.

\begin{theorem}
\label{th-Sp-2}
Let the inhomogeneities $\xi _{1},\ldots,\xi _{N}$ be generic \eqref{cond-inh}. 
Let us suppose that  $\bar{\mathsf{b}}_{-}\not=0$, and in addition that $\bar{\mathsf{c}}_{+}=0$.
Then 
$t(\lambda )\in \Sigma _{\mathcal{T}}$  if and only if  there exists a unique function $Q_{t}(\lambda)$ of the form
\begin{equation}\label{form-Q}
   Q_{t}(\lambda )=\prod_{b=1}^{q}\left( \lambda ^{2}-\lambda_{b}^{2}\right) ,
   \quad q\le N,
   \quad
   \lambda_1,\ldots,\lambda_q \in\mathbb{C}\setminus \big\{ \pm \xi_1^{(0)},\ldots,\pm \xi_N^{(0)} \big\},
\end{equation}
such that 
\begin{equation}
  t(\lambda )\, Q_{t}(\lambda )
=\mathsf{A}_{\bar{\zeta}_{+},\bar{\zeta}_{-}}(\lambda)\, Q_{t}(\lambda -\eta )
+\mathsf{A}_{\bar{\zeta}_{+},\bar{\zeta}_{-}}(-\lambda)\, Q_{t}(\lambda +\eta ).  \label{Eq-homo-1}
\end{equation}
Similarly, $t(\lambda )\in \Sigma _{\mathcal{T}}$ if and only if  there exists a unique function $P_t(\lambda)$ of the form
\begin{equation}\label{form-barQ}
   P_{t}(\lambda )=\prod_{b=1}^{p}\left( \lambda ^{2}-\mu_{b}^{2}\right) ,
   \quad p\le N,
   \quad
   \mu_1,\ldots,\mu_{p}\in\mathbb{C}\setminus \big\{ \pm \xi_1^{(0)},\ldots,\pm \xi_N^{(0)} \big\},
\end{equation}
such that 
\begin{equation}
t(\lambda )\, P_{t}(\lambda )
=\mathsf{A}_{-\bar{\zeta}_{+},-\bar{\zeta}_{-}}(\lambda )\, P_{t}(\lambda -\eta )
+\mathsf{A}_{-\bar{\zeta}_{+},-\bar{\zeta}_{-}}(-\lambda )\, P_{t}(\lambda +\eta ).  \label{Eq-homo-2}
\end{equation}
\end{theorem}

\begin{proof}[Proof]
The proof of this statement has become standard in the SoV approach \cite{NicT15}. It can be derived for example adapting to the present case the steps of the proof given recently in \cite{KitMNT16} for the antiperiodic XXX spin-1/2 quantum chain.
\end{proof}

It is interesting to remark that the two  polynomials $Q_{t}(\lambda )$ and $P_{t}(\lambda )$ of Theorem~\ref{th-Sp-2} 
satisfy some (generalized) wronskian equation, which in particular fixes their total degree, as stated in the following proposition.

\begin{proposition}\label{prop-wronskian}
Under the same hypothesis as in Theorem~\ref{th-Sp-2}, let $Q_t(\lambda)$ and $P_t(\lambda)$ be the two polynomials \eqref{form-Q} and \eqref{form-barQ} associated to a given $t(\lambda )\in \Sigma _{\mathcal{T}}$, {\it i.e.} such that the equations \eqref{Eq-homo-1} and \eqref{Eq-homo-2} are satisfied.
Then
\begin{equation}\label{W-eq}
W_{Q_t,P_t}(\lambda )
=2 (-1)^N \big[{\bar{\zeta}}_++{\bar{\zeta}}_-+(p-q)\eta\big]\,
(\lambda-\eta/2)\, a(\-\lambda)\, d(\lambda),
\end{equation}
where we have defined:
\begin{multline}
  W_{Q_t,P_t}(\lambda ) 
  = \Big(\lambda-\frac{\eta}{2}+\bar{\zeta}_+\Big)\Big(\lambda-\frac{\eta}{2}+\bar{\zeta}_-\Big)\,
  Q_t(\lambda -\eta ) \, P_t(\lambda )
    \\
   - 
   \Big(\lambda-\frac{\eta}{2}-\bar{\zeta}_+\Big)\Big(\lambda-\frac{\eta}{2}-\bar{\zeta}_-\Big)\,
   Q_t(\lambda )\, P_t(\lambda -\eta ). \label{Wronsk}
\end{multline}
It follows that
\begin{equation}
  p+q=N. \label{QbarQ-degree}
\end{equation}
\end{proposition}

\begin{proof}
The two equations \eqref{Eq-homo-1} and \eqref{Eq-homo-2} of Theorem~\ref{th-Sp-2} admit the following rewriting:
\begin{multline}
\bar{\zeta}_+\,\bar{\zeta}_-\, \lambda \,  t(\lambda ) \, Q_{t}(\lambda )
 =\mathsf{A}(\lambda) \, \Big(\lambda-\frac{\eta}{2}+\bar{\zeta}_+\Big)\Big(\lambda-\frac{\eta}{2}+\bar{\zeta}_-\Big)\, Q_{t}(\lambda -\eta )
 \\
  -\mathsf{A}(-\lambda ) \, \Big(-\lambda-\frac{\eta}{2}+\bar{\zeta}_+\Big)\Big(-\lambda-\frac{\eta}{2}+\bar{\zeta}_-\Big)\, Q_{t}(\lambda +\eta ),
\end{multline}
\begin{multline}
\bar{\zeta}_+\,\bar{\zeta}_-\, \lambda \,  t(\lambda ) \, P_t(\lambda )
 =\mathsf{A}(\lambda) \,  \Big(\lambda-\frac{\eta}{2}-\bar{\zeta}_+\Big)\Big(\lambda-\frac{\eta}{2}-\bar{\zeta}_-\Big)\, P_{t}(\lambda -\eta )
 \\
 -\mathsf{A}(-\lambda) \,  \Big(-\lambda-\frac{\eta}{2}-\bar{\zeta}_+\Big)\Big(-\lambda-\frac{\eta}{2}-\bar{\zeta}_-\Big)\, P_{t}(\lambda +\eta ),
\end{multline}
where we have defined
\begin{equation}
  \mathsf{A}(\lambda ) =(-1)^N\, ( \lambda +\eta /2)\, a(\lambda)\, d(-\lambda ).
\end{equation}
Multiplying the first and the second equation respectively by $P_{t}(\lambda )$ and $Q_{t}(\lambda )$ and taking their difference we get that
\begin{equation}
\mathsf{A}(\lambda ) \, W_{Q_t,P_t}(\lambda )
=\mathsf{A}(-\lambda ) \, W_{Q_t,P_t}(-\lambda ).
\end{equation}
We now use the fact that $\mathsf{A}(\lambda )$ and $W_{Q_t,P_t}(\lambda )$ are both polynomial functions in $\lambda$, and that none of the roots of $\mathsf{A}(\lambda )$ coincide with a root of $\mathsf{A}(-\lambda )$.
Hence
\begin{equation}
W_{Q_t,P_t}(\lambda )=w(\lambda)\, \mathsf{A}(-\lambda )
\end{equation}
where $w(\lambda)$ is an even polynomial function of $\lambda$.
Noticing moreover that $\mathsf{A}(\lambda-\eta/2)=-\mathsf{A}(-\lambda-\eta/2)$ and that $W_{Q_t,P_t}(\lambda+\eta/2 )=-W_{Q_t,P_t}(\lambda )(-\lambda+\eta/2)$, we get that $w(\lambda+\eta/2)=w(-\lambda+\eta/2)$, so that $w(\lambda)$ is in fact a constant.
The value of this constant can be fixed by comparing the leading asymptotic behavior of $W_{Q_t,P_t}(\lambda )$ and of $\mathsf{A}(-\lambda)$ when $\lambda\to \pm \infty$:
\begin{align}
  &W_{Q_t,P_t}(\lambda )\sim \left[ 2{\bar{\zeta}}_+ + 2{\bar{\zeta}}_-+2(p-q)\eta\right] \lambda^{2p+2q+1},
  \\
  &\mathsf{A}(-\lambda)\sim -\lambda^{2N+1},
\end{align}
from which it also follows that $p+q=N$.
\end{proof}

Note finally that, similarly to the antiperiodic case \cite{NicT15,KitMNT16}, it is possible to reformulate the characterization of the transfer matrix spectrum of Theorem~\ref{th-Sp-2} by using the above wronskian equation instead of the functional $T$-$Q$ equations.

\begin{proposition}
\label{th-Sp-3}
Under the same hypothesis as in Theorem~\ref{th-Sp-2}, $t(\lambda )\in \Sigma _{\mathcal{T}}$  if and only if there exist two polynomials $Q_t(\lambda )$ and $P_t(\lambda )$ of the respective form \eqref{form-Q} and \eqref{form-barQ} such that
\begin{multline}
t(\lambda )
 =\Big[2 \lambda \big({\bar{\zeta}}_++{\bar{\zeta}}_-+(p-q)\eta\big)\Big]^{-1} 
   \sum_{\epsilon=\pm} \epsilon\Big(\epsilon \lambda-\frac{\eta}{2}+\bar{\zeta}_+\Big)
        \Big(\epsilon\lambda+\frac{\eta}{2}+\bar{\zeta}_+\Big)
 \\
 \times     \Big(\epsilon\lambda-\frac{\eta}{2}+\bar{\zeta}_-\Big) \Big(\epsilon\lambda+\frac{\eta}{2}+\bar{\zeta}_-\Big)\,
         Q_t(\epsilon\lambda-\eta)\, P_t(\epsilon\lambda+\eta),
\end{multline}
and such that $Q_t(\lambda )$ and $P_t(\lambda )$ satisfy the equation \eqref{Wronsk}.
\end{proposition}

\subsection{Bethe ansatz form of the transfer matrix eigenstates 
}

Let us define the following states in the SoV basis,
\begin{align}
   &\ket{ \Omega}
   =\frac{1}{N_{\boldsymbol{\xi},-}}
   \sum_{\mathbf{h}\in\{0,1\}^N}
     \widehat{V}\big( \xi_1^{(h_1)},\ldots,\xi_N^{(h_N)} \big) \, \ket{\mathbf{h}_- },
     \label{Omega}\\
  &\bra{ \Omega_L}
   =\frac{1}{N_{\boldsymbol{\xi},-}}
     \sum_{\mathbf{h}\in\{0,1\}^N}
     \prod_{a=1}^{N} ( f_{a}\, g_{a} )^{h_{a}}\,
     \widehat{V}\big( \xi_1^{(h_1)},\ldots,\xi_N^{(h_N)} \big) \, \bra{\mathbf{h}_- },
     \label{Omega_g} 
\end{align}
and
\begin{align}
  &\bra{ \underline{\Omega}}
  =\frac{1}{N_{\boldsymbol{\xi},-}}
  \sum_{\mathbf{h}\in\{0,1\}^N}
     \prod_{a=1}^{N}  f_{a} ^{h_{a}}\
     \widehat{V}\big( \xi_1^{(h_1)},\ldots,\xi_N^{(h_N)} \big) \, \bra{\mathbf{h}_- },
  \label{barOmega}\\
  &\ket{ \underline{\Omega}_R}
  =\frac{1}{N_{\boldsymbol{\xi},-}}
  \sum_{\mathbf{h}\in\{0,1\}^N}
     \prod_{a=1}^{N}  g_{a} ^{-h_{a}}\
     \widehat{V}\big( \xi_1^{(h_1)},\ldots,\xi_N^{(h_N)} \big) \, \ket{\mathbf{h}_- }.
     \label{barOmega_g}
\end{align}
where $f_a$ and $g_a$ are respectively given by \eqref{f_n} and \eqref{g_n},
and the following renormalized operator
\begin{equation}\label{B-renorm}
\mathcal{B}(\lambda )\equiv (-1)^N\frac{\bar{\mathcal{B}}_{-}(\lambda )}{{b}_{-}(\lambda)}.
\end{equation}
Then, the transfer matrix eigenstates that we have previously constructed by means of SoV can be rewritten in a form similar to what is obtained in the ABA framework, {\it i.e.} through the multiple action of the operator \eqref{B-renorm} on the states \eqref{Omega}-\eqref{barOmega_g}.

\begin{proposition}
\label{th-ABA-1}
Let the inhomogeneities $\xi _{1},\ldots,\xi _{N}$ be generic \eqref{cond-inh}.
Then, for any $t(\lambda )\in\Sigma _{\mathcal{T}}$, the corresponding right and left one-dimensional $\mathcal{T}(\lambda)$-eigenspaces are respectively generated by the following vectors:
\begin{equation}\label{Bethe-vect1}
  \Gamma_W^{-1}\,\prod_{a=1}^{q}\mathcal{B}(\lambda _{a})\, \ket{\Omega }
   =
     \frac{\prod_{k=1}^q d(\lambda_k)\, d(-\lambda_k)}{\prod_{k=1}^p d(\mu_k)\, d(-\mu_k)}\,
         \Gamma_W^{-1}\,  \prod_{a=1}^{p}\mathcal{B}(\mu_{a}) \, \ket{\underline{\Omega}_R },
\end{equation}
and
\begin{equation}\label{Bethe-vect2}
   \bra{\Omega_L } \prod_{a=1}^{q}\mathcal{B}(\lambda _{a}) \, \Gamma_W\, 
  = \frac{\prod_{k=1}^q d(\lambda_k)\, d(-\lambda_k)}{\prod_{k=1}^p d(\mu_k)\, d(-\mu_k)}\,
      \bra{ \underline{\Omega}} \prod_{a=1}^{p}\mathcal{B}(\mu_{a}) \, \Gamma_W \, . 
\end{equation}
In these expressions, $Q_t$ (with roots given by $\lambda_1,\ldots,\lambda_q$) and $P_t$ (with roots given by $\mu_1,\ldots,\mu_p$) are respectively the solutions \eqref{Q-form} of \eqref{Inhom-BAX-eq} and \eqref{P-form} of \eqref{Inhom-BAX-eq-bis} if $\bar{\mathsf{c}}_{+}\not=0$  (in which case $q=p=N$), or the solutions \eqref{form-Q} of \eqref{Eq-homo-1} and \eqref{form-barQ} of \eqref{Eq-homo-2} if  $\bar{\mathsf{c}}_{+}=0$ (in which case $p+q=N$).
\end{proposition}

\begin{proof}[Proof]
The eigenstates \eqref{eigenT-right} and \eqref{eigenT-left} can be written in terms of the zeros of $Q_t$,
\begin{equation}
    \ket{\Psi_t}=N_{\boldsymbol{\xi},-}
   ^{-1}\, \,\Gamma_W^{-1}\prod_{a=1}^{q}\mathcal{B}(\lambda _{a}) \,  \ket{\Omega }, 
    \qquad
    \bra{\Psi_t}=N_{\boldsymbol{\xi},-}\,
    \bra{\Omega_L } \prod_{a=1}^{q}\mathcal{B}(\lambda _{a}) \, \Gamma_W,
\end{equation}
using the fact the left and right SoV-basis \eqref{left-basis} and \eqref{right-basis} are by definition eigenbasis of $\bar{\mathcal{B}}_{-}(\lambda )$ with eigenvalue \eqref{eigen-B}, and that
\begin{equation}
   \prod_{a=1}^q\prod_{n=1}^N\big(\lambda_a-\xi_n^{(h_n)}\big)\big(-\lambda_a-\xi_n^{(h_n)}\big)
   = \prod_{n=1}^N Q_t\big( \xi_n^{(h_n)}\big).\label{Rewriting-SoV-ABA-Id}
\end{equation}
These states can also be expressed in terms of the zeros of $P_t$, by using the reference states $\ket{\underline\Omega_R }$ and $\bra{\underline\Omega}$ instead of $\ket{\Omega }$ and $\bra{\Omega_L }$:
\begin{align}
    &\ket{\Psi_t}=N_{\boldsymbol{\xi},-}\,
    \prod_{a=1}^{N}\frac{Q_t(\xi_{a}^{(0)})}{P_t(\xi_{a}^{(0)})} \,
       \, \Gamma_W^{-1} \prod_{a=1}^{p}\mathcal{B}(\mu_{a}) \, \ket{\underline{\Omega}_R },
        \\
    &\bra{\Psi_t}=N_{\boldsymbol{\xi},-}\,
    \prod_{a=1}^{N}\frac{Q_t(\xi_{a}^{(0)})}{P_t(\xi_{a}^{(0)})} \,
      \bra{ \underline{\Omega}} \prod_{a=1}^{p}\mathcal{B}(\mu_{a})\, \Gamma_W .
\end{align}
This follows from the fact that
\begin{align}
\frac{P_t(\xi _{n}^{(1)})}{P_t(\xi _{n}^{(0) })}
 &=\frac{\mathsf{A}_{\bar{\zeta}_{+},\bar{\zeta}_{-}}(\xi _{n}^{(0)})}{\mathsf{A}_{-\bar{\zeta}_{+},-\bar{\zeta}_{-}}(\xi _{n}^{(0)})}\,
  \frac{Q_t(\xi _{n}^{( 1) })}{Q_t(\xi _{n}^{( 0) })}
   =g_n \, \frac{Q_t(\xi _{n}^{( 1) })}{Q_t(\xi _{n}^{( 0) })}.
\end{align}
Finally, noticing that
\begin{equation}
     \prod_{a=1}^{N}\frac{Q_t(\xi_{a}^{(0)})}{P_t(\xi_{a}^{(0)})} 
     = \frac{\prod_{k=1}^q d(\lambda_k)\, d(-\lambda_k)}{\prod_{k=1}^p d(\mu_k)\, d(-\mu_k)},
\end{equation}
we get the result.
\end{proof}

It is worth remarking that, under the condition $\bar{\mathsf{c}}_{+}= 0$, it is also possible to use the Bethe ansatz approach to characterize the transfer matrix spectrum. However, the completeness of the corresponding  ABA construction of eigenstates has so far remained a complicated issue. Proposition~\ref{th-ABA-1}, based on the SoV study of the model, 
provides instead a {\em complete} construction for the eigenstates as ``Bethe'' vectors of the form \eqref{Bethe-vect1}-\eqref{Bethe-vect2}.
The connection with the usual ABA construction becomes clearer if we remark that the reference state \eqref{Omega} actually coincides with the state with all spin up, whereas the reference state \eqref{barOmega} coincides with the states with all spin down, as stated in the following proposition.

\begin{proposition}
\label{th-ABA-2}
Let the inhomogeneities $\xi _{1},\ldots,\xi _{N}$ be generic  \eqref{cond-inh} and ${\bar{\mathsf{b}}}_{-}\neq 0$.
Then
\begin{equation}
   \ket{\Omega}=\ket{0},\qquad
   \bra{ \underline{\Omega}}=\bra{\underline{0}} ,  \label{SOV-rewriting-reference}
\end{equation}
where $ \ket{0 } $ and  $\bra{ \underline{0} }$ are respectively the  right reference states with all spin up and the left reference state with all spin down.
\end{proposition}

\begin{proof}[Proof]
Let us first remark that the construction of the SoV basis of Appendix~\ref{app-Beigen} relies only on the reflection algebra for ${\mathcal{U}}_-$, and hence does not depend on the boundary matrix $\bar{K}_+$ \eqref{triangKnobar}. It follows that the reference states $\ket{\Omega}$ \eqref{Omega} and $\bra{\underline{\Omega}}$ \eqref{barOmega} remain the same if we vary $\bar{\mathsf{c}}_+$ without varying $\bar{K}_-$.

Let us now remark that, if we fix $\bar{\mathsf{c}}_+$ to be zero, we
have the following expression for the transfer matrix $ \mathcal{{T}}(\lambda )$:
\begin{align}
   \bar{\mathcal{T}}(\lambda ) 
   &=\frac{(2\lambda +\eta )\, (\lambda-\frac{\eta}{2}+\bar{\zeta}_+ )\, \bar{\mathcal{A}}_{-}(\lambda )
                +(2\lambda -\eta) \,(-\lambda-\frac{\eta}{2}+\bar{\zeta}_+ )\, \bar{\mathcal{A}}_{-}(-\lambda )}
               {2\lambda\, \bar{\zeta}_+ }, 
               \\
   &=\frac{(2\lambda +\eta ) \,(-\lambda+\frac{\eta}{2}+\bar{\zeta}_+ )\,  \bar{\mathcal{D}}_{-}(\lambda )
                +(2\lambda -\eta ) \,  (\lambda+\frac{\eta}{2}+\bar{\zeta}_+ )\, \bar{\mathcal{D}}_{-}(-\lambda )}{2\lambda\, \bar{\zeta}_+ }.
\end{align}
From the boundary-bulk decomposition of the elements of the gauged
transformed boundary monodromy matrix,
\begin{align*}
&\bar{\mathcal{A}}_{-}(\lambda )
 =(-1)^N \Big\{ \bar{a}_{-}(\lambda )\, A(\lambda )\, D(-\lambda ) -\bar{b}_{-}(\lambda )\, A(\lambda) \, C(-\lambda )
 -\bar{d}_{-}(\lambda ) \, B(\lambda ) \, C(-\lambda ) \Big\}, 
 \\
 &\bar{\mathcal{D}}_{-}(\lambda )
  =(-1)^N  \Big\{ \bar{d}_{-}(\lambda ) \, D(\lambda ) \, A(-\lambda )
  -\bar{a}_{-}(\lambda ) \, C(\lambda) \, B(-\lambda )
  +\bar{b}_{-}(\lambda ) \, C(\lambda ) \, A(-\lambda ) \Big\},
\end{align*}
it follows that
\begin{align}
  &\bra{ \underline{0}} \, \bar{\mathcal{D}}_{-}(\lambda ) 
   =(-1)^N \bar{d}_{-}(\lambda ) \, a(\lambda) \, d(-\lambda )\, \bra{ \underline{0}},
    \\
  &\bar{\mathcal{A}}_{-}(\lambda ) \, \ket{0} 
  = (-1)^N  \bar{a}_{-}(\lambda )\, a(\lambda )\, d(-\lambda)\, \ket{0},
\end{align}
where we have used the following properties: 
\begin{alignat}{3}
  & A(\lambda )\, \ket{0}=a(\lambda )\, \ket{0} ,\qquad 
  & D(\lambda) \, \ket{0} =d(\lambda ) \, \ket{0} ,\qquad 
  & C(\lambda ) \, \ket{0} =0,
  \\
  &\bra{\underline{0}}\, A(\lambda ) =d(\lambda ) \, \bra{ \underline{0}},\qquad  
  &\bra{\underline{0}} \, D(\lambda )=a(\lambda ) \, \bra{\underline{0}},\qquad  
  &\bra{\underline{0}} \, C(\lambda )=0.
\end{alignat}
These identities imply that
\begin{align}
  &\bra{\underline{0}}\,\bar{\mathcal{T}}(\lambda )
   =\left(\mathsf{A}_{-\bar{\zeta}_{+},-\bar{\zeta}_{-}}(\lambda )+\mathsf{A}_{-\bar{\zeta}_{+},-\bar{\zeta}_{-}}(-\lambda )\right)
   \bra{\underline{0}}, \\
  &\bar{\mathcal{T}}(\lambda )\, \ket{0}
   =\left(\mathsf{A}_{\bar{\zeta}_{+},\bar{\zeta}_{-}}(\lambda )+\mathsf{A}_{\bar{\zeta}_{+},\bar{\zeta}_{-}}(-\lambda )\right)
   \ket{0},
\end{align}
on the gauge transformed transfer matrix, which together with the spectrum
simplicity imply that $\ket{\Omega}$ is proportional to $\ket{0}$, whereas $\bra{\underline{\Omega}}$ is proportional to $\bra{\underline{0}}$.
It is then easy to show that the proportionality factor is equal to $1$ by computing the scalar product of these right (respectively left) vectors with $\bra{0}\equiv \bra{\mathbf{1}_-}$ (respectively with $\ket{\underline{0}}\equiv \ket{\mathbf{0}_-}$).

Since, from the above remark,  $\ket{\Omega}$ and $\bra{\underline{\Omega}}$ do not depend on $\bar{K}_+$, the equality remains valid also if we do not have $\bar{\mathsf{c}}_+=0$.
\end{proof}

\section{Scalar products of separate states 
}
\label{sec-sp}

As usual  \cite{GroMN12,Nic13}, the separation of variable method enables one to obtain determinant representations for the scalar products of the so-called {\em separate states}.
By separate state, we mean here a right (respectively left) state that can be written in the following form,
\begin{align}
    &\sum_{\mathbf{h}\in\{0,1\}^N}\prod_{n=1}^{N} \gamma_{n}^{(h_n)}\
   \widehat{V}\big( \xi_1^{(h_1)},\ldots,\xi_N^{(h_N)} \big) \, \Gamma_W^{-1}  \,  \ket{\mathbf{h}_-} ,
   \label{separate-r} \\
   &\sum_{\mathbf{h}\in\{0,1\}^N}
    \prod_{n=1}^{N} \gamma_{n}^{(h_n)}\
   \widehat{V}\big( \xi_1^{(h_1)},\ldots,\xi_N^{(h_N)} \big)    \,\bra{\mathbf{h}_-}
    \, \Gamma_W,
    \label{separate-l}
\end{align}
for some coefficients $\gamma_n^{(h_n)}$, $n\in\{1,\ldots,N\}$, $h_n\in\{0,1\}$.
This class of states is of particular interest since, on the one hand, it contains the complete set of transfer matrix eigenstates \eqref{eigenT-right} and \eqref{eigenT-left} and, on the other hand, it can be used as a convenient generating family of the space of states\footnote{The fact that we do not have any constraint on the coefficients $\gamma_n^{(h_n)}$ makes it {\em a priori} much easier to express any state as a linear combination of separate states than  as a linear combination of the transfer matrix eigenstates themselves. In particular, we expect the action of local operators on the eigenstates to be explicitly computable in terms of linear combinations of separate states, as in  \cite{GroMN12,KitMNT16}.}. In this respect, these separate states somehow play the same role as the {\em off-shell} Bethe vectors in the context of algebraic Bethe ansatz. Therefore we expect the computation of their scalar products to be a fundamental step towards the computation of more complicated physical quantities such as form factors and correlation functions as was shown in \cite{GroMN12} in the SoV context.

In fact, it is easy to see that any arbitrary separate state, as defined in \eqref{separate-r} or \eqref{separate-l}, admits representations as some Bethe-type vector. To clarify this point, let us introduce some notations that we shall use throughout the remaining part of the paper. Given a polynomial function in $\lambda^2$,
\begin{equation}
\beta (\lambda )=\prod_{m=1}^{n_{\beta }}(\lambda ^{2}-b_{m}^{2}),
\label{poly-sep-state}
\end{equation}
we define, by multiple action of the operator \eqref{B-renorm} on the reference states \eqref{Omega}-\eqref{barOmega_g}, the following Bethe-type vectors:
\begin{align}
& \ket{\beta }=\Gamma _{W}^{-1}\prod_{a=1}^{n_{\beta }}\mathcal{B}(b_{a})\,\ket{\Omega },
  \label{Bethe1}\\
& \bra{ \underline{\beta} }=\bra{ \underline{\Omega}}\prod_{a=1}^{n_{\beta }}\mathcal{B}(b_{a})\,\Gamma _{W},  \label{Bethe2}\\
& \bra{\beta }=\bra{ \Omega_L }\prod_{a=1}^{n_{\beta }}\mathcal{B}(b_{a})\,\Gamma _{W} 
  \label{Bethe3}\\
& \ket{\underline{\beta} }=\Gamma_{W}^{-1}\prod_{a=1}^{n_{\beta }}\mathcal{B}(b_{a})\,\ket{\underline{\Omega}_R },  \label{Bethe4}
\end{align}
Then the separate vector \eqref{separate-r} is equal to the Bethe-type vector \eqref{Bethe1} (respectively to the Bethe-type vector \eqref{Bethe4}) if the polynomial function $\beta(\lambda)$ is chosen such that $\beta(\xi_n^{(h_n)})=N^{1/N}_{\boldsymbol{\xi},-}\,\gamma_n^{(h_n)},$ (respectively such that $\beta(\xi_n^{(h_n)})=N^{1/N}_{\boldsymbol{\xi},-}\,\gamma_n^{(h_n)}\, g_n^{h_n}$) for all $n\in\{1,\ldots,N\}$, $h_n\in\{0,1\}$.
Similarly,  the separate vector \eqref{separate-l} is equal to the Bethe-type vector \eqref{Bethe2} (respectively to the Bethe-type vector \eqref{Bethe3}) if the polynomial function $\beta(\lambda)$ is chosen such that 
$\beta(\xi_n^{(h_n)})=N^{1/N}_{\boldsymbol{\xi},-}\,\gamma_n^{(h_n)}\, f_n^{-h_n},$ 
(respectively such that $\beta(\xi_n^{(h_n)})=N^{1/N}_{\boldsymbol{\xi},-}\,\gamma_n^{(h_n)}\, (f_n\,g_n)^{-h_n}$) for all $n\in\{1,\ldots,N\}$, $h_n\in\{0,1\}$. Conversely, any Bethe-type vector of the form \eqref{Bethe1} or \eqref{Bethe4} (respectively of the form \eqref{Bethe2} or \eqref{Bethe3}) is obviously a separate state of the form \eqref{separate-r} (respectively \eqref{separate-l}).
Hence, in the present article, we shall from now on use the notations \eqref{poly-sep-state}-\eqref{Bethe4} instead of \eqref{separate-r}-\eqref{separate-l}.

In the previous section, we have shown that the one-dimensional right $\mathcal{T}(\lambda )$-eigenspace associated with the eigenvalue $t(\lambda )
$ is generated by any of the two (proportional) separate states 
\begin{equation}
\ket{ Q_t }=\Gamma _{W}^{-1}\prod_{a=1}^{q}\mathcal{B}(\lambda _{a})\,%
\ket{\Omega},\qquad \ket{ \underline{{P}_{t}} }=\Gamma
_{W}^{-1}\prod_{a=1}^{p}\mathcal{B}(\mu _{a})\,\ket{\underline{\Omega}_R}.
\label{eigen2}
\end{equation}
whereas the one-dimensional left $\mathcal{T}(\lambda )$-eigenspace
associated with the eigenvalue $t(\lambda )$ is generated by any of the two
(proportional) separate states,
\begin{equation}
\bra{ Q_t}=\bra{ \Omega_L }\prod_{a=1}^{q}\mathcal{B}(\lambda _{a})\,\Gamma
_{W},\qquad \bra{ \underline{{P}_t }}=\bra{ \underline{\Omega}}%
\prod_{a=1}^{p}\mathcal{B}(\mu _{a})\,\Gamma _{W}.  \label{eigen1}
\end{equation}
In these expressions, we have used the above Bethe-type notations to represent the eigenstates as some on-shell particularization of \eqref{Bethe1}-\eqref{Bethe4}. $Q_{t}$ and $P_{t}$ stand here for the polynomials \eqref{Q-form} and \eqref{P-form} with $p=q=N$ satisfying respectively \eqref{Inhom-BAX-eq} and \eqref{Inhom-BAX-eq-bis} if $\bar{\mathsf{c}}_{+}\not=0$, and for the polynomials \eqref{form-Q} and \eqref{form-barQ} with $p+q=N$ satisfying respectively \eqref{Eq-homo-1} and \eqref{Eq-homo-2} if $\bar{\mathsf{c}}_{+}=0$.

\subsection{The scalar product of two generic separate states
}

The previous SoV studies  \cite{GroMN12,Nic13} applied here for the XXX open chain enables one to straightforwardly express the scalar product of two generic separate states as defined in \eqref{separate-r} and \eqref{separate-l} and hence in \eqref{poly-sep-state} in the form of a determinant. 

\begin{proposition}\label{prop-sc1}
 Let $\alpha(\lambda)$ and $\beta(\lambda)$ be two polynomials in $\lambda^2$ of respective degree $n_\alpha$ and $n_\beta$, and which can be expressed as
\begin{equation}
   \alpha(\lambda)=\prod_{k=1}^{n_\alpha}(\lambda^2-\alpha_k^2),
   \qquad
   \beta(\lambda)=\prod_{k=1}^{n_\beta}(\lambda^2-\beta_k^2),
\end{equation}
in terms of some sets of roots $\{\alpha_1^2,\ldots,\alpha_{n_\alpha}^2\}$ and $\{\beta_1^2,\ldots,\beta_{n_\beta}^2\}$.
Then the scalar products of the separate states built as in \eqref{separate-r} and \eqref{separate-l}  from $\alpha(\lambda)$ and $\beta(\lambda)$,
\begin{equation}\label{id-sp1}
   \moy{\alpha\, |\, \beta}
   =\moy{\beta\,|\,\alpha}
   =
   \bra{\Omega_L}\prod_{k=1}^{n_\alpha}\mathcal{B}(\alpha_k)\prod_{k=1}^{n_\beta}\mathcal{B}(\beta_k)
   \ket{\Omega},
\end{equation}
admit the following determinant representations:
\begin{align}
  \moy{\alpha\, |\, \beta}
  & = \frac{1}{N_{\boldsymbol{\xi},-}}\,
  \det_{1\le i,j\le N}
 \left[ \sum_{h=0}^1 (f_i\, g_i)^h\, \alpha(\xi_i^{(h)})\, \beta(\xi_i^{(h)})\, \big(\xi_i^{(h)}\big)^{2(j-1)}\right],
  \label{repr-sp1}\\
  &=  \prod_{n=1}^N\frac{\xi_n-{\bar{\zeta}}_-}{\xi_n\,{\bar{\mathsf{b}}}_{-}}\,
  \frac{ \det_{1\le i,j\le N}
  \left[ \sum_{h=0}^1 (-g_i)^h \alpha(\xi_i^{(h)})\, \beta(\xi_i^{(h)})\, \big(\xi_i^{(1-h)}\big)^{2(j-1)}\right]}
  {\widehat{V}(\xi_1,\ldots,\xi_N)}.
  \label{repr-sp2}
\end{align}
We also have
\begin{align}\label{id-sp1bis}
   \moy{\underline{\alpha}\, |\, \beta}
   &=\moy{\beta\,|\,\underline{\alpha}}
   =\moy{\alpha\, |\, \underline{\beta}}
   =\moy{\underline{\beta}\, |\, \alpha}
   =  \bra{\underline{\Omega}}\prod_{k=1}^{n_\alpha}\mathcal{B}(\alpha_k)\prod_{k=1}^{n_\beta}\mathcal{B}(\beta_k)
   \ket{\Omega},
  \\
  & = \frac{1}{N_{\boldsymbol{\xi},-}}\,
  \det_{1\le i,j\le N}
  \left[ \sum_{h=0}^1 f_i^h\, \alpha(\xi_i^{(h)})\, \beta(\xi_i^{(h)})\, \big(\xi_i^{(h)}\big)^{2(j-1)}\right]
  \label{repr-sp1bis}\\
  &= \prod_{n=1}^N\frac{\xi_n-{\bar{\zeta}}_-}{\xi_n\,{\bar{\mathsf{b}}}_{-}}\,
  \frac{\det_{1\le i,j\le N}
  \left[ \sum_{h=0}^1 (-1)^h \alpha(\xi_i^{(h)})\, \beta(\xi_i^{(h)})\, \big(\xi_i^{(1-h)}\big)^{2(j-1)}\right]}
  {\widehat{V}(\xi_1,\ldots,\xi_N)}.
  \label{repr-sp2bis}
\end{align}
Finally,
\begin{align}\label{id-sp1ter}
   \moy{\underline{\alpha}\, |\, \underline{\beta}}
  & =\moy{\underline{\beta}\,|\,\underline{\alpha}}
   =\bra{\underline{\Omega}}\prod_{k=1}^{n_\alpha}\mathcal{B}(\alpha_k)\prod_{k=1}^{n_\beta}\mathcal{B}(\beta_k)
   \ket{\underline{\Omega}_R},
  \\
  & =  \frac{1}{N_{\boldsymbol{\xi},-}}\,
  \det_{1\le i,j\le N}
  \left[ \sum_{h=0}^1 \big(f_i\, g_i^{-1}\big)^h\, \alpha(\xi_i^{(h)})\, \beta(\xi_i^{(h)})\, \big(\xi_i^{(h)}\big)^{2(j-1)}\right],
  \label{repr-sp1ter}\\
  &= \prod_{n=1}^N\frac{\xi_n-{\bar{\zeta}}_-}{\xi_n\,{\bar{\mathsf{b}}}_{-}}\,
  \frac{\det_{1\le i,j\le N}
  \left[ \sum_{h=0}^1 \big(-g_i^{-1}\big)^h \alpha(\xi_i^{(h)})\, \beta(\xi_i^{(h)})\, \big(\xi_i^{(1-h)}\big)^{2(j-1)}\right]}
  {\widehat{V}(\xi_1,\ldots,\xi_N)}.
  \label{repr-sp2ter}
\end{align}
\end{proposition}

\begin{proof}
The fact that the scalar products in \eqref{id-sp1},  \eqref{id-sp1bis} or \eqref{id-sp1ter} coincide is a simple consequence of the commutativity of the family of operators $\mathcal{B}(\lambda)$ and of the definitions of the reference states. 
The representation \eqref{repr-sp1} follows from the fact that the separate states $\bra{\alpha}$ and $\ket{\beta}$ can be rewritten in terms of the left and right SoV basis \eqref{left-SOV-state} and \eqref{right-SOV-state} as
\begin{align}
   &\bra{\alpha}=  \frac{1}{N_{\boldsymbol{\xi},-}}\,
    \sum_{\mathbf{h}\in\{0,1\}^N}
   \prod_{n=1}^N\left[ (f_n\, g_n)^{h_n}\, \alpha(\xi_n^{(h_n)})\right]\, 
   \widehat{V}\big( \xi_1^{(h_1)},\ldots,\xi_N^{(h_N)} \big) 
   \,\bra{\mathbf{h}_-}\, \Gamma_W
  , 
   \\
   &\ket{\beta}=  \frac{1}{N_{\boldsymbol{\xi},-}}\, 
    \sum_{\mathbf{h}\in\{0,1\}^N}\prod_{n=1}^N\beta(\xi_n^{(h_n)})\,
   \widehat{V}\big( \xi_1^{(h_1)},\ldots,\xi_N^{(h_N)} \big) 
\, \Gamma_W^{-1} \, \ket{\mathbf{h}_-},
\end{align}
from the orthogonality relation \eqref{norm} and from the multi linearity of the determinant with respect to the sum over $h_1,\ldots,h_N\in\{0,1\}$.
The representation \eqref{repr-sp2} can be obtained similarly using the identity \eqref{id-VDM} so as to rewrite the state $\bra{\alpha}$ as
\begin{equation*}
  \bra{\alpha}= \prod_{n=1}^N\frac{\xi_n-{\bar{\zeta}}_-}{\xi_n\,{\bar{\mathsf{b}}}_{-}}
 \!\! \sum_{\mathbf{h}\in\{0,1\}^N} \prod_{n=1}^N \left[ (-g_n)^{h_n}\, \alpha(\xi_n^{(h_n)})\right]\, 
   \frac{ \widehat{V}\big( \xi_1^{(1-h_1)},\ldots,\xi_N^{(1-h_N)} \big) }
          { \widehat{V}\big( \xi_1,\ldots,\xi_N \big) }
   \,\bra{\mathbf{h}_-}\, \Gamma_W
 .
\end{equation*}
The representations \eqref{repr-sp1bis}, \eqref{repr-sp2bis}, \eqref{repr-sp1ter} and \eqref{repr-sp2ter} are obtained similarly.
\end{proof}

\begin{rem}
The orthogonality condition for two different transfer matrix eigenstates, which can for instance be written as
\begin{equation}
   \moy{\underline{{Q}_{t'}}\, |\, P_t}=\moy{P_{t'}\, |\, \underline{{Q}_t}}=0\qquad \text{if } t'\not= t,
\end{equation}
can be seen directly at the level of the determinant representation \eqref{repr-sp1bis}, by showing that the corresponding matrix admits a non-zero eigenvector. Indeed, following the method used in \cite{GroMN12}, the difference of two eigenvalues $t(\lambda)$ and $t'(\lambda)$ divided by $(4\lambda^2-\eta^2)$ is a non-zero polynomials in $\lambda^2$ of degree $N-1$,
\begin{equation}
   t(\lambda)-t'(\lambda)=\left(4\lambda^2-\eta^2\right) \, \sum_{j=1}^N C_j^{(t,t')}\, \lambda^{2(j-1)},
\end{equation}
which satisfies  
\begin{multline}
   \sum_{j=1}^N  \left[ \sum_{h=0}^1 f_i^h\, Q_{t'}(\xi_i^{(h)})\, P_t(\xi_i^{(h)})\, \big(\xi_i^{(h)}\big)^{2(j-1)}\right] C_j^{(t,t')}\\
    = Q_{t'}(\xi_i^{(0)})\, P_t(\xi_i^{(0)})\, \frac{t(\xi_i^{(0)})-t'(\xi_i^{(0)})}{4(\xi_i^{(0)})^2-\eta^2}
    +f_i\, Q_{t'}(\xi_i^{(1)})\, P_t(\xi_i^{(1)})\, \frac{t(\xi_i^{(1)})-t'(\xi_i^{(1)})}{4(\xi_i^{(1)})^2-\eta^2}.
\end{multline}
Thanks to the $T$-$Q$ equations \eqref{Eq-homo-1} and \eqref{Eq-homo-2}, the right hand side reduces to
\begin{multline}
   \frac{ Q_{t'}(\xi_i^{(0)})\, P_t(\xi_i^{(1)})\, \mathsf{A}_{-\bar{\zeta}_+,-\bar{\zeta}_-}(\xi_i^{(0)})
   - Q_{t'}(\xi_i^{(1)})\, P_t(\xi_i^{(0)})\, \mathsf{A}_{\bar{\zeta}_+,\bar{\zeta}_-}(\xi_i^{(0)}) }{4(\xi_i^{(0)})^2-\eta^2}
   \\
   +f_i\, \frac{ Q_{t'}(\xi_i^{(1)})\, P_t(\xi_i^{(0)})\, \mathsf{A}_{-\bar{\zeta}_+,-\bar{\zeta}_-}(-\xi_i^{(1)})
   - Q_{t'}(\xi_i^{(0)})\, P_t(\xi_i^{(1)})\, \mathsf{A}_{\bar{\zeta}_+,\bar{\zeta}_-}(-\xi_i^{(1)}) }{4(\xi_i^{(1)})^2-\eta^2},
\end{multline}
which is equal to zero due to the identities
\begin{equation*}
   \frac{\mathsf{A}_{-\bar{\zeta}_+,-\bar{\zeta}_-}(\xi_i^{(0)})}{4(\xi_i^{(0)})^2-\eta^2}
   =f_i\, \frac{ \mathsf{A}_{-\bar{\zeta}_+,-\bar{\zeta}_-}(-\xi_i^{(1)})}{4(\xi_i^{(1)})^2-\eta^2},
   \qquad
   \frac{\mathsf{A}_{\bar{\zeta}_+,\bar{\zeta}_-}(\xi_i^{(0)}) }{4(\xi_i^{(0)})^2-\eta^2}
   =f_i\, \frac{\mathsf{A}_{-\bar{\zeta}_+,-\bar{\zeta}_-}(-\xi_i^{(1)})}{4(\xi_i^{(1)})^2-\eta^2}.
\end{equation*}
\end{rem}

The representations for the scalar products of separate states that we have obtained in Proposition~\ref{prop-sc1} directly from the SoV study of the model are, as usual within the framework of this approach \cite{GroMN12,GroMN14,Nic12,Nic13,Nic13a,Nic13b,FalKN14,LevNT16}, not properly formulated for the consideration of the homogeneous limit. 
As shown in the remaining part of this section, it is possible to transform them into more convenient ones for the consideration of both homogeneous and thermodynamic limits.
These reformulations are based on  some algebraic determinant identities quite similar as those that were used in \cite{KitMNT16} in the case of the XXX antiperiodic chain.

\subsection{Some useful determinant identities}
\label{subsec-id}

We prove here several determinant identities that we shall use in the next subsection to reformulate the representations \eqref{repr-sp2},  \eqref{repr-sp2bis} and  \eqref{repr-sp2ter} for the scalar product of separate states.

So as to concisely formulate these identities and the resulting quantities, let us first introduce some notations.
For a set of arbitrary variables $\{x\}\equiv\{x_1,\ldots,x_L\}$ and a function $f$, we define
\begin{equation}\label{def-A}
   \mathcal{A}_{\{x\}}[f]
   =
   \frac{\det_{1\le i,j \le L}\left[ \sum_{\epsilon\in\{+,-\}} 
     f(\epsilon x_i)
     \left(x_i+\epsilon\frac{\eta}{2}\right)^{2(j-1)}\right]}
          {\det_{1\le i,j \le L}\left[ x_i^{2(j-1)}\right]},
\end{equation}
We also define, for  two arbitrary parameters $\xi_+$, $\xi_-$ and  a set of arbitrary variables $\{z\}\equiv\{z_1,\ldots,z_M\}$, a function
\begin{equation}\label{f-zeta-z}
    f_{\xi_+,\xi_-,\{z\}}(\lambda)=\frac{(\lambda+\xi_+)(\lambda+\xi_-)}{\lambda}\prod_{m=1}^M\frac{\lambda^2-z_m^2}{(\lambda+\frac{\eta}{2})^2-z_m^2}.
\end{equation}
We finally define, when the two sets of variable $\{x\}$ and $\{z\}$ have the same cardinality $M=L$,
\begin{equation}
     \mathcal{I}_{\xi_+,\xi_-}(\{ x \},\{z\})
    =\frac{\prod_{j,k=1}^L(x_j^2-z_k^2)}
              {\prod_{j<k}(x_j^2-x_k^2)(z_k^2-z_j^2)}\,
      \det_{1\le i,j \le L}\left[  \mathcal{I}_{\xi_+,\xi_-} (x_i,z_j) \right] ,
\label{Izergin-det}
\end{equation}
where the function $\mathcal{I}_{\xi_+,\xi_-} (x,z)$ is given by
\begin{equation}\label{mat-I}
    \mathcal{I}_{\xi_+,\xi_-} (x,z)
    =\sum_{\epsilon\in\{+,-\}} \!\!\!
     \epsilon\frac{(x+\epsilon\xi_+)(x+\epsilon\xi_-)}{x \left[ (x+\epsilon\frac{\eta}{2})^2-z^2\right] }.
\end{equation}
These functions can be seen as the boundary counterparts of the one appearing in Section 3.2 of \cite{KitMNT16}. In particular, $\mathcal{I}_{\xi_+,\xi_-}(\{ x \},\{z\})$ can be understood as a generalization of the Izergin's determinant \cite{Ize87,Tsu98}.

\begin{identity}\label{prop-id-1}
Let $M= L$. Then the above functions are related by the following identities:
\begin{align}
    \mathcal{A}_{\{z\}}\big[f_{\xi_+,\xi_-,\{x\}}\big]
      &= \mathcal{I}_{\xi_+,\xi_-}(\{ z \},\{x\})\label{MainK-id-1}\\
      &= (-1)^L \ \mathcal{I}_{-\xi_++\frac{\eta}{2},-\xi_-+\frac{\eta}{2}}(\{ x \},\{z\})\label{MainK-id-2}\\
      &=(-1)^L \, \mathcal{A}_{\{x\}}\big[f_{-\xi_++\frac{\eta}{2},-\xi_-+\frac{\eta}{2},\{z\}}\big].\label{MainK-id-3}
\end{align}
\end{identity}

\begin{proof}
We consider, associated with the set of variables $\{x\}\equiv\{x_1,\ldots,x_L\}$, the following polynomials in $\lambda^2$,
\begin{equation}\label{poly-X}
  X_k(\lambda)=\prod_{\substack{\ell=1\\ \ell\not=k}}^L\!\left(\lambda^2-x_\ell^2\right),
  \qquad 1\le k\le L,
\end{equation}
and the $L\times L$ matrix $\mathcal{C}^X$ with elements $\mathcal{C}_{j,k}^{X}$, $1\le j,k\le L$, defined by the relations
\begin{equation}\label{mat-CX}
     X_k(\lambda ) 
     =\sum_{j=1}^{L} \mathcal{C}_{j,k}^{X}\, \lambda ^{2(j-1)},
     \qquad
     \forall k\in \{ 1,\ldots ,L \} .  
\end{equation}
It is easy to compute its determinant by observing that
\begin{align}
      \det_{1\le i,j \le L}\left[ x_i^{2(j-1)}\right]\cdot \det_{1\le j,k \le L}\left[ \mathcal{C}^{X}_{j,k}\right]
     &=\det_{1\le i,k\le L}\left[ \sum_{j=1}^{L} \mathcal{C}_{j,k}^{X}\, x_i^{2(j-1)}\right] \nonumber\\
     &=\det_{1\le i,k \le L}\left[ X_k(x_{i})\, \delta _{i,k}\right] ,
\end{align}
so that
\begin{equation}
    \det_{L}\left[ \mathcal{C}^{X}\right] 
    =\frac{\prod_{\ell=1}^{L}X_{\ell}(x_{\ell})}{\prod_{1\leq a<b\leq L}(x_b^2-x_a^2)}
    =\prod_{1\leq b<a\leq L}(x_b^2-x_a^2)
    =\widehat{V}(x_L,\ldots, x_1).
\end{equation}
Moreover, it is simple to compute its product with the $L\times L$ matrix of elements $\sum_{\epsilon\in\{+,-\}} 
     f_{\xi_+,\xi_-,\{x\}}(\epsilon z_i)
     \left(z_i+\epsilon\frac{\eta}{2}\right)^{2(j-1)}$, by remarking that, for $1\le i,k \le L$,
\begin{align}
    \sum_{j=1}^{L}
    \Bigg[ \sum_{\epsilon\in\{+,-\}} \!\! & f_{\xi_+,\xi_-,\{x\}}(\epsilon z_i)
     \left(z_i+\epsilon\frac{\eta}{2}\right)^{2(j-1)} \Bigg]\, \mathcal{C}_{j,k}^{X}
            \nonumber\\
    &\quad
    = \!\! \sum_{\epsilon\in\{+,-\}} \!\!
      f_{\xi_+,\xi_-,\{x\}}(\epsilon z_i)\,
     \sum_{j=1}^{L} \mathcal{C}_{j,k}^{X} \left(z_i+\epsilon\frac{\eta}{2}\right)^{2(j-1)}
     \nonumber\\
     &\quad
    =\prod_{\ell=1}^L\!\left(z_i^2-x_\ell^2\right)
       \sum_{\epsilon\in\{+,-\}} \!\!\!
     \epsilon\frac{(z_i+\epsilon\xi_+)(z_i+\epsilon\xi_-)}{z_i \left[ (z_i+\epsilon\frac{\eta}{2})^2-x_k^2\right] },
      \label{prod-mat1}
\end{align}
which implies \eqref{MainK-id-1}.
Rewriting now $\mathcal{I}_{\xi_+,\xi_-} (z,x)$ as
\begin{equation}
     \mathcal{I}_{\xi_+,\xi_-} (z,x)
    = 
    \frac{\sum_{\epsilon=\pm} \epsilon (z+\epsilon\xi_+)(z+\epsilon\xi_-)\left[(z-\epsilon\frac{\eta}{2})^2-x^2\right]}
           {z \left[ (z+\frac{\eta}{2})^2-x^2\right] \left[(z-\frac{\eta}{2})^2-x^2\right]},
\end{equation}
and noticing that
\begin{equation}\label{id-1}
   \frac{1}{ \left[ (z+\frac{\eta}{2})^2-x^2 \right]  \left[ (z-\frac{\eta}{2})^2-x^2\right] }
 =\frac{1}{ \left[ (x+\frac{\eta}{2})^2-z^2 \right] \left[ (x-\frac{\eta}{2})^2-z^2\right]},
\end{equation}
and that
\begin{align}\label{id-2}
  &\frac{1}{z}\sum_{\epsilon=\pm} \epsilon (z+\epsilon\xi_+)(z+\epsilon\xi_-)\Big[\left(z-\epsilon\frac{\eta}{2}\right)^{\! 2}-x^2\Big]
  \nonumber\\
  &\qquad
  = 2\Big[ z^2(\xi_++\xi_--\eta)-x^2(\xi_++\xi_-)+ \frac{\eta^2}{4}(\xi_++\xi_-)-\eta\xi_+\xi_-\Big]
  \nonumber\\
  &\qquad
  = 2\Big[ x^2(\tilde\xi_++\tilde\xi_--\eta)-z^2(\tilde\xi_++\tilde\xi_-)+ \frac{\eta^2}{4}(\tilde\xi_++\tilde\xi_-)-\eta\tilde\xi_+\tilde\xi_-\Big]
  \nonumber\\
  &\qquad
  = \frac{1}{x}\sum_{\epsilon=\pm} \epsilon (x+\epsilon\tilde\xi_+)(x+\epsilon\tilde\xi_-)\Big[\left(x-\epsilon\frac{\eta}{2}\right)^{\! 2}-z^2\Big],
\end{align}
in which we have set $\tilde\xi_\pm=-\xi_\pm+\frac{\eta}{2}$, we obtain \eqref{MainK-id-2}.
Finally \eqref{MainK-id-3} is a consequence of the previous identities.
\end{proof}

The identity between $\mathcal{A}_{\{z\}}\big[f_{\xi_+,\xi_-,\{x\}}\big]$ and $\mathcal{A}_{\{x\}}\big[f_{-\xi_++\frac{\eta}{2},-\xi_-+\frac{\eta}{2},\{z\}}\big]$ can in fact easily be extended to the case in which the cardinality of the two sets of variables is different.

\begin{identity}\label{prop-id-2}
Let us consider two sets of arbitrary complex numbers $\{x\}\equiv\{x_1,\ldots,x_L\}$ and $\{z\}\equiv\{z_1,\ldots,z_M\}$, and let us suppose that $M \le  L$. Then, for any complex parameters $\xi_+$, $\xi_-$, we have
\begin{equation}
    \mathcal{A}_{\{x\}}\big[f_{\xi_+,\xi_-,\{z\}}\big]
    =(-1)^M\, 2^{L-M}\!\!\prod_{j=0}^{L-M-1}\!\!\!\!(\xi_++\xi_-+j\eta)\
    \mathcal{A}_{\{z\}}\big[f_{-\xi_++\frac{\eta}{2},-\xi_-+\frac{\eta}{2},\{x\}}\big].
    \label{MainK-id-4}
\end{equation}
\end{identity}

\begin{proof}
For $M=L$ the formula is identical to the previous one \eqref{MainK-id-3}.  For $M<L$ it is easy to see that
\begin{align}
   &\mathcal{A}_{\{x_1,\ldots,x_L\}}\big[f_{\xi_+,\xi_-,\{z_1,\ldots,z_M\}}\big]
      \nonumber\\
   &\qquad
    =\lim_{z_{M+1}\to +\infty}\ldots\lim_{z_{L}\to +\infty}
      \mathcal{A}_{\{x_1,\ldots,x_L\}}\big[f_{\xi_+,\xi_-,\{z_1,\ldots,z_L\}}\big]
   \nonumber\\
   &\qquad
    =(-1)^L \lim_{z_{M+1}\to +\infty}\ldots\lim_{z_{L}\to +\infty}
   \mathcal{A}_{ \{z_1,\ldots,z_L\}} \big[ f_{-\xi_++\frac{\eta}{2},-\xi_-+\frac{\eta}{2},\{x_1,\ldots,x_L\}}\big],
   \label{id-limits-1}
\end{align}
in which we have used Identity~\ref{prop-id-1}.
Computing the limits in \eqref{id-limits-1}, we finally obtain \eqref{MainK-id-4}.
\end{proof}

The last identity has an immediate corollary in the particular case where $M<L$ and where $\xi_++\xi_-+j\eta=0$ for some $j\in\{0,1,\ldots,L-M-1\}$:

\begin{corollary}\label{prop-id-2-1}
With the same notations as in Identity~\ref{prop-id-2}, let us suppose that $M<L$ and that there exists some integer $j\in\{0,1,\ldots,L-M-1\}$ such that $\xi_++\xi_-+j\eta=0$. Then
\begin{equation}
   \mathcal{A}_{\{x_1,\ldots,x_L\}}\big[f_{\xi_+,\xi_-,\{z_1,\ldots,z_M\}}\big]
    =0.
\end{equation}
\end{corollary}

Instead, if $\xi_+$ and $\xi_-$ are generic, we can rewrite Identity~\ref{prop-id-2} in a form which does not depend on which of the two sets $\{x\}$ or $\{z\}$ has bigger cardinality:

\begin{corollary}\label{prop-id-2-2}
Let $\{x\}\equiv\{x_1,\ldots,x_L\}$ and  $\{z\}\equiv\{z_1,\ldots,z_M\}$ be two sets of arbitrary complex numbers. Then, for any values of the parameters  $\xi_+$ and $\xi_-$ such that 
the corresponding ratio of $\Gamma$-functions is well defined, we have
\begin{multline}
     \mathcal{A}_{\{x\}}\big[f_{\xi_+,\xi_-,\{z\}}\big]
    =(-1)^M\, (2\eta)^{L-M} \, \\
    \times
    \frac{\Gamma\big(\frac{\xi_++\xi_-}{\eta}+L-M\big)}{\Gamma\big(\frac{\xi_++\xi_-}{\eta}\big)}\
    \mathcal{A}_{\{z\}}\big[f_{-\xi_++\frac{\eta}{2},-\xi_-+\frac{\eta}{2},\{x\}}\big].\label{MainK-id-5}
\end{multline}
\end{corollary}

\begin{rem}
Identity~\ref{prop-id-2}, Corollary~\ref{prop-id-2-1} and Corollary~\ref{prop-id-2-2} enable us to compute the quantity $\mathcal{A}_{\{x\}}\big[f_{\xi_+,\xi_-,\{z\}}\big]$, or to explicitly rewrite it in terms of $\mathcal{A}_{\{z\}}\big[f_{-\xi_++\frac{\eta}{2},-\xi_-+\frac{\eta}{2},\{x\}}\big]$ for all values of  $\xi_+$ and $\xi_-$ such that $\frac{\xi_++\xi_-}{\eta} \notin\{1,\ldots,M-L\}$ (the latter set being non-empty only for $M>L$). For the special values for which the relation \eqref{MainK-id-5} is not well-defined, it is possible instead to explicitly rewrite  $\mathcal{A}_{\{x\}}\big[f_{\xi_+,\xi_-,\{z\}}\big]$ in terms of some modified version of $\mathcal{A}_{\{z\}}\big[f_{-\xi_++\frac{\eta}{2},-\xi_-+\frac{\eta}{2},\{x\}}\big]$.
However, since $\mathcal{A}_{\{x\}}\big[f_{\xi_+,\xi_-,\{z\}}\big]$ is a polynomial function of $\xi_+$ and $\xi_-$, one can also merely use analytic continuation at these special points.
\end{rem}

Finally, it is possible to rewrite, for any function $f$, the ratio of determinants \eqref{def-A} in terms of a generalized version of the usual scalar products determinant representations \cite{Sla89}.
For a one-variable function $f$ and two sets of variables  $\{x\}\equiv\{x_1,\ldots,x_M\}$ and $\{y\}\equiv\{y_1,\ldots,y_L\}$ such that $L\ge M$, we define
\begin{equation}\label{def-gen-Slav-det}
    \mathcal{S}_{\{x\},\{y\}}[f]  
    =     \frac{\widehat{V}(x_1-\frac{\eta}{2},\ldots,x_M-\frac{\eta}{2})}{\widehat{V}(x_1+\frac{\eta}{2},\ldots,x_M+\frac{\eta}{2})}\,
       \frac{\det_L\mathcal{S}_{\mathbf{x},\mathbf{y}}[f]}{\widehat{V}(x_M,\ldots,x_1)\, \widehat{V}(y_1,\ldots,y_L)},
\end{equation}
where  $\mathcal{S}_{\mathbf{x},\mathbf{y}}[f]$ is the $L\times L$ matrix with elements
\begin{multline}\label{def-gen-Slav-mat}
   \big[\, \mathcal{S}_{\mathbf{x},\mathbf{y}}[f] \, \big]_{i,k}
   =\!\!\sum_{\epsilon\in\{+,-\}} \!\!
      f(\epsilon y_i)\ 
     X(y_i+\epsilon\eta)
     \\
     \times
     \begin{cases}
      {\displaystyle
      \, \left[ \frac{f(-x_k)}{(y_i+\epsilon\frac{\eta}{2})^2-(x_k+\frac{\eta}{2})^2}
     -\frac{f(x_k)\ \varphi_{\{x\}}(x_k)}{(y_i+\epsilon\frac{\eta}{2})^2-(x_k-\frac{\eta}{2})^2}\right]
     }
     \quad &\text{if } k\le M,
     \vspace{2mm}\\
    \, \left(y_i+\epsilon\frac{\eta}{2}\right)^{\! 2(k-M-1)}
     \qquad &\text{if } k >M.
     \end{cases}
\end{multline}
Here we have used the notations
\begin{equation}
   \varphi_{\{x\}}(\lambda)
   =\frac{2\lambda-\eta}{2\lambda+\eta}\,
     \frac{X(\lambda+\eta)}{X(\lambda-\eta)}
     \quad \text{and}\quad
     X(\lambda)=\prod_{m=1}^M(\lambda^2-x_m^2).
\end{equation}

\begin{identity}\label{prop-id-3}
Let $\{x\}\equiv\{x_1,\ldots,x_M\}$ and $\{y\}\equiv\{y_1,\ldots,y_L\}$ be two sets of arbitrary complex numbers, and let us suppose that $L\ge M$.
Then, for any function $f$,
\begin{equation}
       \mathcal{A}_{ \{x\}\cup\{y\} }[f]
       =\mathcal{S}_{\{x\},\{y\}}[f].
\end{equation}
\end{identity}

The proof of this identity is given in Appendix~\ref{app-id}.

\subsection{The scalar product of two generic separate states: alternative determinant representations}

We can now use the determinant identities obtained in the previous subsection so as to reformulate the determinant representations obtained in Proposition~\ref{prop-sc1} for the scalar products of two generic separate states.
As stated in the following theorem, we notably obtain a generalization of the representations \cite{Sla89,FodW12} already valid at the off-shell level, {\it i.e.} without requiring any of the two states to be a transfer matrix eigenstate.

\begin{theorem}\label{th-sp-separate-gen}
With the same notations as in Proposition~\ref{prop-sc1}, and supposing that $n_\beta\ge n_\alpha$, the scalar product of the separate states $\bra{\alpha}$ and $\ket{\beta}$ can be expressed as
\begin{align}
   \moy{\alpha\, |\, \beta} 
   &= 
     \frac{(2\eta)^{ N-n_\alpha-n_\beta}}{\prod_{n=1}^N\big[({\bar{\zeta}}_+-\xi_n)\, {\bar{\mathsf{b}}}_-\big]}
    \frac{\Gamma\big(\frac{\bar{\zeta}_++\bar{\zeta}_-}{\eta}\! +\! N\! -\! n_\alpha\! -\! n_\beta\big)}{\Gamma\big(\frac{\bar{\zeta}_++\bar{\zeta}_-}{\eta}\big)} \
  \mathcal{A}_{ \{\alpha\}\cup\{\beta\} }\big[\, \widetilde{\mathsf{A}}_{{\bar{\zeta}}_+,{\bar{\zeta}}_-}\big],
       \label{sc-A-1}\\
 \nonumber\\
 & = 
    \frac{(2\eta)^{ N-n_\alpha-n_\beta}}{\prod_{n=1}^N\big[({\bar{\zeta}}_+-\xi_n)\, {\bar{\mathsf{b}}}_-\big]} \,
    \frac{\Gamma\big(\frac{\bar{\zeta}_++\bar{\zeta}_-}{\eta}\! +\! N\! -\! n_\alpha\! -\! n_\beta\big)}{\Gamma\big(\frac{\bar{\zeta}_++\bar{\zeta}_-}{\eta}\big)}\ 
    \mathcal{S}_{\{\alpha\},\{\beta\}}\big[\,\widetilde{\mathsf{A}}_{{\bar{\zeta}}_+,{\bar{\zeta}}_-}\big],
       \label{sc-S-1}
\end{align}
in which we have used the notation \eqref{def-A}, \eqref{def-gen-Slav-det}, and
\begin{equation}
  \widetilde{\mathsf{A}}_{\bar{\zeta},\bar{\zeta}'}(-\lambda)
  =\frac{(\lambda-\frac{\eta}{2}+\bar{\zeta})(\lambda-\frac{\eta}{2}+\bar{\zeta}')}{\lambda}\,
  a(\lambda)\, d(-\lambda).
\end{equation}
Similarly, the scalar product of the separate states $\bra{\underline\alpha}$ and $\ket{\underline\beta}$ can be expressed as
\begin{align}
   \moy{\underline\alpha\, |\, \underline\beta} 
   &= 
    \prod_{n=1}^N\frac{\xi_n-{\bar{\zeta}}_-}{({\bar{\zeta}}_+-\xi_n)\, (\xi_n+{\bar{\zeta}}_-)\,{\bar{\mathsf{b}}}_-}\,
    (2\eta)^{\! N-n_\alpha-n_\beta}\,
    \frac{\Gamma\big(\!-\!\frac{\bar{\zeta}_++\bar{\zeta}_-}{\eta}\! +\! N\! -\! n_\alpha\! -\! n_\beta\big)}{\Gamma\big(\! -\! \frac{\bar{\zeta}_++\bar{\zeta}_-}{\eta}\big)} 
      \nonumber\\
  &\hspace{7.6cm} \times 
  \mathcal{A}_{ \{\alpha\}\cup\{\beta\} }\big[\, \widetilde{\mathsf{A}}_{-{\bar{\zeta}}_+,-{\bar{\zeta}}_-}\big],
      \label{sc-A-2}\\
 \nonumber\\
 & =    \prod_{n=1}^N\frac{\xi_n-{\bar{\zeta}}_-}{({\bar{\zeta}}_+-\xi_n)\, (\xi_n+{\bar{\zeta}}_-)\,{\bar{\mathsf{b}}}_-}\,
    (2\eta)^{\! N-n_\alpha-n_\beta}\, 
    \frac{\Gamma\big(\!-\!\frac{\bar{\zeta}_++\bar{\zeta}_-}{\eta}\! +\! N\! -\! n_\alpha\! -\! n_\beta\big)}{\Gamma\big(\! -\! \frac{\bar{\zeta}_++\bar{\zeta}_-}{\eta}\big)}
      \nonumber\\
  & \hspace{7.6cm}
  \times
        \mathcal{S}_{\{\alpha\},\{\beta\}}\big[\,\widetilde{\mathsf{A}}_{-{\bar{\zeta}}_+,-{\bar{\zeta}}_-}\big].
       \label{sc-S-2}
\end{align}
Finally, the scalar product  of the separate states $\bra{\underline\alpha}$ and $\ket{\beta}$ vanishes if $n_\alpha+n_\beta<N$. For $n_\alpha+n_\beta\ge N$, it is equal to
\begin{align}
   \moy{\underline\alpha\, |\, \beta} 
   &= 
   \prod_{n=1}^N\frac{\xi_n-{\bar{\zeta}}_-}{\big(\bar{\zeta}^2-\xi_n^2\big)\,{\bar{\mathsf{b}}}_-}\,
   \frac{(2\eta)^{N-n_\alpha-n_\beta}}{(n_\alpha+n_\beta-N)!}\   \mathcal{A}_{ \{\alpha\}\cup\{\beta\} }\big[\, \widetilde{\mathsf{A}}_{\bar{\zeta},-\bar{\zeta}}\big],
   \label{sc-A-3}\\
   &= \prod_{n=1}^N\frac{\xi_n-{\bar{\zeta}}_-}{\big(\bar{\zeta}^2-\xi_n^2\big)\,{\bar{\mathsf{b}}}_-}\, \frac{(2\eta)^{N-n_\alpha-n_\beta}}{(n_\alpha+n_\beta-N)!}\
      \mathcal{S}_{\{\alpha\},\{\beta\}}\big[\,\widetilde{\mathsf{A}}_{\bar{\zeta},-\bar{\zeta}}\big],
       \label{sc-S-3}
\end{align}
in which $\bar{\zeta}$ is an arbitrary parameter.
\end{theorem}

\begin{proof}
From \eqref{repr-sp2}, we can express $\moy{\alpha\, |\, \beta}$ using the functions \eqref{def-A} and \eqref{f-zeta-z} introduced in the previous subsection:
\begin{equation}\label{sc-interm-1}
    \moy{\alpha\, |\, \beta}
    =(-1)^N 
    \prod_{n=1}^N\frac{\alpha(\xi_n^{(0)})\, \beta(\xi_n^{(0)})\, \alpha(\xi_n^{(1)})\, \beta(\xi_n^{(1)})}{(\xi_n-{\bar{\zeta}}_+)\,{\bar{\mathsf{b}}}_-\, \alpha(\xi_n)\, \beta(\xi_n)}\
       \mathcal{A}_{\{\xi\}}\big[f_{{\bar{\zeta}}_+,{\bar{\zeta}}_-,\{\alpha\}\cup\{\beta\}}\big].
\end{equation}
This enables us to apply Corollary~\ref{prop-id-2-2}. We get
\begin{multline}
    \moy{\alpha\, |\, \beta}
    =(-1)^{N+n_\alpha+n_\beta}\, (2\eta)^{N-n_\alpha-n_\beta}\, 
    \prod_{n=1}^N\frac{\alpha(\xi_n^{(0)})\, \beta(\xi_n^{(0)})\, \alpha(\xi_n^{(1)})\, \beta(\xi_n^{(1)})}{(\xi_n-{\bar{\zeta}}_+)\,{\bar{\mathsf{b}}}_-\,  \alpha(\xi_n)\, \beta(\xi_n)}
      \\
   \times 
   \frac{\Gamma\big(\frac{\bar{\zeta}_++\bar{\zeta}_-}{\eta}+N-n_\alpha-n_\beta\big)}{\Gamma\big(\frac{\bar{\zeta}_++\bar{\zeta}_-}{\eta}\big)}\,
   \mathcal{A}_{\{\alpha\}\cup\{\beta\}}\big[f_{-{\bar{\zeta}}_++\frac{\eta}{2},-{\bar{\zeta}}_-+\frac{\eta}{2},\{\xi\}}\big].
\end{multline}
Reintroducing some of the pre-factors into the determinant, we get \eqref{sc-A-1}. \eqref{sc-S-1} is then a direct consequence of Identity~\ref{prop-id-3}.

The scalar product $\moy{\underline\alpha\, |\, \underline\beta} $ can be obtained similarly.

Finally, the scalar product $\moy{\underline\alpha\, |\, \beta} $ can be expressed as in \eqref{sc-interm-1} in terms of $ \mathcal{A}_{\{\xi\}}\big[f_{\bar{\zeta},-\bar{\zeta},\{\alpha\}\cup\{\beta\}}\big]$ for any arbitrary parameter $\bar{\zeta}$. From Identity~\ref{prop-id-2} and Corollary~\ref{prop-id-2-1}, the latter is equal to
\begin{equation}
   \mathcal{A}_{\{\xi\}}\big[f_{\bar{\zeta},-\bar{\zeta},\{\alpha\}\cup\{\beta\}}\big]
   =\begin{cases}
      0  &\text{if } n_\alpha+n_\beta<N,\\
      {\displaystyle \frac{(-1)^N\, \mathcal{A}_{\{\alpha\}\cup\{\beta\}}\big[f_{-\bar{\zeta}+\frac{\eta}{2},\bar{\zeta}+\frac{\eta}{2} ,\{\xi\}}\big]}{(2\eta)^{n_\alpha+n_\beta-N} (n_\alpha+n_\beta-N) !} }
        \quad &\text{if } n_\alpha+n_\beta\ge N,
      \end{cases}
\end{equation}
which leads to the result.
\end{proof}

\subsection{The scalar product of a generic separate state with an eigenstate of the transfer matrix}

We would like to stress once again that the determinant representations of Theorem~\ref{th-sp-separate-gen} are valid for any separate states, and that we did not need to suppose that one of them is an eigenstate of the transfer matrix to obtain them. In other words, these formulas do not rely on the use of particular Bethe equations, they are just a natural consequence of the SoV construction of the space of states. Moreover, they are already in an adequate form for the study of the homogeneous and thermodynamic limits, and we can of course substitute any of the set of variables by a set of solutions of the Bethe equations (homogeneous or inhomogeneous) to study them in the case of eigenstates.

It is nevertheless worth noticing that, when $\bar{\mathsf{c}}_{+}=0$ and that one of the separate state is an eigenstate of the transfer matrix, the representations \eqref{sc-S-1} and \eqref{sc-S-2} slightly simplify, so that we obtain a direct generalization of the usual  determinant representation for the scalar products of Bethe states that are known in the context of algebraic Bethe ansatz \cite{Sla89,KitMT99,Wan02,KitKMNST07}.

Let us consider a separate state  built as in \eqref{separate-r} or \eqref{separate-l} from a $\lambda^2$-polynomial $\beta(\lambda)$ \eqref{poly-sep-state} of degree $n_\beta$, and an eigenstate of the transfer matrix associated with an eigenvalue $t(\lambda)$. We recall that this eigenstate can be built as in \eqref{eigen1} or \eqref{eigen2} either from the polynomial $Q_t(\lambda)$ \eqref{form-Q} of degree $q$ and roots $\lambda_1^2,\ldots,\lambda_q^2$ which solves \eqref{Eq-homo-1}, or from the polynomial $P_t(\lambda)$ of degree $p=N-q$ and roots $\mu_1^2,\ldots,\mu_p^2$ which solves \eqref{Eq-homo-2}.

\begin{theorem}
\label{th-charact-SP1}
Let the inhomogeneity parameters $\xi _{1},\ldots ,\xi _{N}$ be generic  \eqref{cond-inh}
and $\bar{\mathsf{c}}_+=0$.
Then, the scalar product $\moy{Q_t\, |\, \beta}=\moy{\beta\, |\, Q_t}$ vanishes if $n_\beta<q$. When $n_\beta\ge q$, it can be expressed as
\begin{multline}
  \moy{\beta\, |\, Q_t}
  =  \frac{(2\eta)^{ N-q-n_\beta}}{\prod_{n=1}^N\big[({\bar{\zeta}}_+-\xi_n)\,  {\bar{\mathsf{b}}}_-\big]}\,
    \frac{\Gamma\big(\frac{\bar{\zeta}_++\bar{\zeta}_-}{\eta}\! +\! N\! -\! q\! -\! n_\beta\big)}{\Gamma\big(\frac{\bar{\zeta}_++\bar{\zeta}_-}{\eta}\big)}\,
    \prod_{k=1}^q\widetilde{\mathsf{A}}_{\bar{\zeta}_+,\bar{\zeta}_-}(-\lambda_k)    
     \\
     \times  \prod_{i=1}^{n_\beta}\frac{2(-1)^N \bar{\zeta}_+\bar{\zeta}_-\, Q_t(\beta_i)}{\eta^2-4\beta_i^2}\,
     \frac{\widehat{V}(\lambda_1-\frac{\eta}{2},\ldots,\lambda_q-\frac{\eta}{2})}{\widehat{V}(\lambda_1+\frac{\eta}{2},\ldots,\lambda_q+\frac{\eta}{2})}\,
       \frac{\det_{n_\beta}\mathcal{S}_t(\{\beta\})}{\widehat{V}(\lambda_q,\ldots,\lambda_1)\, \widehat{V}(\beta_1,\ldots,\beta_{n_\beta})},
\end{multline}
where
\begin{equation}\label{Slav-mat-eigen}
   \big[\, \mathcal{S}_t(\{\beta\}) \big]_{i,k}
   =     \begin{cases}
      {\displaystyle
       \frac{\partial\, t(\beta_i)}{\partial\lambda_k}
     }
     \quad &\text{if } k\le q,
     \vspace{2mm}\\
     {\displaystyle
     \sum_{\epsilon\in\{+,-\}} \!\!\! \epsilon\, \mathsf{A}_{\bar{\zeta}_+,\bar{\zeta}_-}(-\epsilon\beta_i)\, 
     \frac{Q_t(\beta_i+\epsilon\eta)}{Q_t(\beta_i)}
     \left(\beta_i+\epsilon\frac{\eta}{2}\right)^{\! 2(k-q)-1}
     }
     \ \ &\text{if } k >q.
     \end{cases}
\end{equation}
The scalar product $\moy{Q_t\, |\,\underline\beta}=\moy{\underline\beta, |\, Q_t}$ vanishes if $n_\beta<N-q$. When $n_\beta\ge N-q=p$, it can be expressed as
\begin{multline}
  \moy{\underline\beta\, |\, Q_t}
  =
    \prod_{n=1}^N\frac{\xi_n-{\bar{\zeta}}_-}{({\bar{\zeta}}_+-\xi_n)\, (\xi_n+{\bar{\zeta}}_-)\,{\bar{\mathsf{b}}}_-}\
    (2\eta)^{\! N-q-n_\beta}\, 
    \frac{\Gamma\big(\!-\!\frac{\bar{\zeta}_++\bar{\zeta}_-}{\eta}\! +\! q\! -\! n_\beta\big)}{\Gamma\big(\!-\!\frac{\bar{\zeta}_++\bar{\zeta}_-}{\eta}\big)}
    \\
    \times
     \frac{\prod_{k=1}^q d(\lambda_k)\, d(-\lambda_k)}{\prod_{k=1}^p d(\mu_k)\, d(-\mu_k)}\
     \prod_{k=1}^{p}\widetilde{\mathsf{A}}_{-\bar{\zeta}_+,-\bar{\zeta}_-}(-\mu_k)\
     \prod_{i=1}^{n_\beta}\frac{2(-1)^N \bar{\zeta}_+\bar{\zeta}_-\,P_t(\beta_i)}{\eta^2-4\beta_i^2}
     \\
     \times
     \frac{\widehat{V}(\mu_1-\frac{\eta}{2},\ldots,\mu_{p}-\frac{\eta}{2})}{\widehat{V}(\mu_1+\frac{\eta}{2},\ldots,\mu_{p}+\frac{\eta}{2})}\,
       \frac{\det_{n_\beta}\bar{\mathcal{S}}_t(\{\beta\})}{\widehat{V}(\mu_{p},\ldots,\mu_1)\, \widehat{V}(\beta_1,\ldots,\beta_{n_\beta})},
\end{multline}
where
\begin{equation}\label{barSlav-mat-eigen}
   \big[\, \bar{\mathcal{S}}_t(\{\beta\}) \big]_{i,k}
   =     \begin{cases}
      {\displaystyle
       \frac{\partial\, t(\beta_i)}{\partial\mu_k}
     }
     \quad &\text{if } k\le p,
     \vspace{2mm}\\
     {\displaystyle
     \sum_{\epsilon\in\{+,-\}} \!\!\! \epsilon\, \mathsf{A}_{-\bar{\zeta}_+,-\bar{\zeta}_-}(-\epsilon\beta_i)\, 
     \frac{P_t(\beta_i+\epsilon\eta)}{P_t(\beta_i)}
     \left(\beta_i+\epsilon\frac{\eta}{2}\right)^{\! 2(k-p)-1}
     }
     \ \ &\text{if } k > p.
     \end{cases}
\end{equation}

\end{theorem}

\begin{proof}[Proof]
For the computation of $\moy{Q_t\, |\, \beta}=\moy{\beta\, |\, Q_t}$ when $n_\beta<q$, we use the fact that
\begin{equation}\label{Q-P}
   \ket{Q_t}=\frac{\prod_{k=1}^q d(\lambda_k)\, d(-\lambda_k)}{\prod_{k=1}^p d(\mu_k)\, d(-\mu_k)}\,
   \ket{\underline{P_t}},
\end{equation}
and that the scalar product $\moy{\beta\, |\, \underline{P_t}}$ vanishes since $p+n_\beta=N-q+n_\beta<N$.
When instead $n_\beta\ge q$, we use the formula \eqref{sc-S-1}. The Bethe equations for  $\lambda_1,\ldots, \lambda_q$   following from \eqref{Eq-homo-1} can be written as
\begin{equation}
   \varphi_{\{\lambda\}}(\lambda_k)=\frac{\widetilde{\mathsf{A}}_{\bar{\zeta}_+,\bar{\zeta}_-}(-\lambda_k)}{\widetilde{\mathsf{A}}_{\bar{\zeta}_+,\bar{\zeta}_-}(\lambda_k)},
   \qquad 1\le k \le q,
\end{equation}
which simplifies the matrix elements in the first $q$th columns of the matrix $\mathcal{S}_{\boldsymbol{\lambda},\boldsymbol{\beta}}$.
Comparing with the derivatives of 
\begin{equation}
    t(\lambda )
=\sum_{\epsilon\in\{+,-\}}\mathsf{A}_{\bar{\zeta}_{+},\bar{\zeta}_{-}}(\epsilon\lambda)\, \frac{Q_{t}(\lambda -\epsilon\eta )}{Q_{t}(\lambda )},
\end{equation}
with respect to $\lambda_k$, we get the result.

The proof is similar for the computation of $\moy{Q_t\, |\, \underline\beta}=\moy{\underline\beta\, |\, Q_t}$: we use \eqref{Q-P} and \eqref{sc-S-2} in the case $n_\beta\ge p$, and the expression of $t(\lambda)$ issued from the functional equation \eqref{Eq-homo-2}.
\end{proof}

\section*{Acknowledgements}
J. M. M., G. N. and V. T. are supported by CNRS. N. K. 
would like to thank LPTHE, University Paris VI, and Laboratoire de Physique, ENS-Lyon for hospitality.


\appendix

\section{Left and right $\mathcal{\bar{B}}_{-}$-eigenbasis}
\label{app-Beigen}


Let us denote with $\bra{ 0}\equiv \otimes _{n=1}^{N}\bra{ \uparrow ,n}$
the dual reference state with all spin up, and by $\ket{\underline{0}}\equiv \otimes
_{n=1}^{N} \ket{\downarrow ,n} $ the reference state with all spin down. They satisfy the properties
\begin{alignat}{4}
  &\bra{0} A(\lambda )=a(\lambda )\bra{0},
  \quad
  &\bra{0} D(\lambda)=d(\lambda )\bra{0},
  \quad 
  &\bra{0} B(\lambda )=0,
  \quad 
  &\bra{0} C(\lambda )\neq 0,\label{prop-bra0}\\
  &A(\lambda )\ket{\underline{0}}=d(\lambda )\ket{\underline{0}},
\quad
  &D(\lambda )\ket{\underline{0}}=a(\lambda )\ket{\underline{0}},
\quad
  &B(\lambda )\ket{\underline{0}}=0,
\quad  
  &C(\lambda)\ket{\underline{0}}\neq 0, \label{prop-bra0bar}
\end{alignat}
where $a(\lambda)$ and $d(\lambda)$ are given by \eqref{a-d}.
It follows from \eqref{prop-bra0}, \eqref{prop-bra0bar} and from the fact that
\begin{multline}\label{Bbord-bulk}
\mathcal{\bar{B}}_{-}(\lambda )= (-1)^N\Big\{-A(\lambda)\, \bar{a}_{-}(\lambda)\, B(-\lambda)
+A(\lambda)\, \bar{b}_{-}(\lambda)\,A(-\lambda)\\
-B(\lambda)\, \bar{c}_{-}(\lambda)\, B(-\lambda)+B(\lambda)\,\bar{d}_-(\lambda)\, A(-\lambda) \Big\},
\end{multline}
that the reference state $\bra{0}$ is a left eigenstate of the operator $\mathcal{\bar{B}}_{-}(\lambda )$ with eigenvalue ${\mathsf{B}}_{-,\mathbf{1}}(\lambda )= (-1)^N\, \bar{b}_-(\lambda)\, a(\lambda)\, a(-\lambda)$, whereas $\ket{\underline{0}}$ is a right eigenstate of $\mathcal{\bar{B}}_{-}(\lambda )$  with eigenvalue $\mathsf{B}_{-,\mathbf{0}}(\lambda)=(-1)^N \bar{b}_-(\lambda)\, d(\lambda)\, d(-\lambda)$.

For each $N$-tuple $\mathbf{h}\equiv(h_1,\ldots,h_N)\in\{0,1\}^N$, we define 
%
\begin{align}
  \label{left-SOV-state}
  &\bra{\mathbf{h}_-}\equiv \bra{0}\prod_{n=1}^{N}
    \left( \frac{\mathcal{\bar{A}}_{-}(\eta /2-\xi _{n})}{\mathsf{A}_{-}(\eta /2-\xi _{n})}\right) ^{1-h_{n}},
    \\
  \label{right-SOV-state}
 &\ket{\mathbf{h}_-} \equiv 
 \prod_{n=1}^{N}\left( \frac{\mathcal{\bar{D}}_{-}(\xi _{n}+\eta /2)}{k_{n}\, \mathsf{A}_{-}(\eta /2-\xi _{n})}\right)
^{h_{n}}\ket{\underline{0}},
\end{align}
where
\begin{equation}
   \mathsf{A}_-(\lambda)=(-1)^N \bar{a}_-(\lambda)\, a(\lambda)\, d(-\lambda),
   \qquad
   k_{n}=\frac{2\xi _{n}+\eta }{2\xi _{n}-\eta }.
\end{equation}

It is easy to see
that, for all  $\mathbf{h}\equiv(h_1,\ldots,h_N)\in\{0,1\}^N$, the states \eqref{left-SOV-state} form a basis of the dual space of states which is a left-eigenbasis of $\mathcal{\bar{B}}_{-}$, whereas the states \eqref{right-SOV-state} form a basis of the space of states which is a right-eigenbasis of $\mathcal{\bar{B}}_{-}$:
\begin{align}
   & \bra{\mathbf{h}_-}\mathcal{\bar{B}}_{-}(\lambda )
    ={\mathsf{B}}_{-,\mathbf{h}}(\lambda )    \bra{\mathbf{h}_-},
    \\
   & \mathcal{\bar{B}}_{-}(\lambda ) \ket{\mathbf{h}_-} ={\mathsf{B}}_{-,\mathbf{h}}(\lambda )\ket{\mathbf{h}_-},
\end{align}
with eigenvalues
\begin{align}
   &{\mathsf{B}}_{-,\mathbf{h}}(\lambda )= (-1)^N\, \bar{b}_-(\lambda)\, 
   \prod_{n=1}^N\big(\lambda-\xi_n^{(h_n)}\big)\big(-\lambda-\xi_n^{(h_n)}\big).
   \label{eigen-B}
\end{align}

The action of $\bar{\mathcal{A}}_-(\lambda)$ and $\bar{\mathcal{D}}_-(\lambda)$ on $\bra{\mathbf{h}_-}$ and on $\ket{\mathbf{h}_-}$ can then be determined as in \cite{Nic12}.
Using the fact that $\bar{\mathcal{A}}_-(\lambda)$ is a polynomial in $\lambda$ of degree $2N+1$ with leading coefficient $\frac{\lambda^{2N+1}}{\bar{\zeta}_-} \mathrm{Id}$, that $\bar{\mathcal{A}}_-(\eta/2)=(-1)^N\mathrm{det}_q M(0)$, and computing the action on $\bra{\mathbf{h}_-}$ of $\bar{\mathcal{A}}_-(\xi_n^{(h_n)})$ and  $\bar{\mathcal{A}}_-(-\xi_n^{(h_n)})$, $1\le n \le N$, we obtain that, for any $\mathbf{h}\in\{0,1\}^N$,
\begin{multline}
  \bra{\mathbf{h}_-}\,\bar{\mathcal{A}}_-(\lambda)
     =\sum_{a=1}^{N}\sum_{\epsilon=\pm 1}
        \frac{ \big(2\lambda -\eta\big)\big(\lambda +\epsilon\xi_{a}^{(h_a)}\big) }{2\xi_{a}^{(h_a)} \big(2\xi_{a}^{(h_a)}-\epsilon\eta \big)}
        \prod_{\substack{ b=1  \\ b\neq a}}^{N}
        \frac{\lambda ^{2}-\big( \xi_{b}^{(h_b)} \big) ^{2}}{\big( \xi_{a}^{(h_a)}\big) ^{2}-\big( \xi_{b}^{(h_b)}\big)^{2}}\\
        \times
        \mathsf{A}_{-}\big( \epsilon\xi_{a}^{(h_a)} \big)\,
       \bra{\mathrm{T}_a^{\epsilon} \mathbf{h}_-}
       \\
       +\Bigg\{(-1)^N \mathrm{det}_{q} M(0)
       \prod_{b=1}^{N}\frac{\lambda ^{2}-\big( \xi_{b}^{(h_b)}\big) ^{2}}
                                            {\big( \eta /2\big) ^{2}-\big( \xi_{b}^{(h_b)}\big) ^{2}}
       +\frac{2\lambda-\eta}{2\bar{\zeta}_-}\prod_{b=1}^\mathsf{N}\Big(\lambda^2-\big(\xi_{b}^{(h_b)}\big)^2\Big)\Bigg\}
        \bra{\mathbf{h}_-},
      \label{act-A}
\end{multline}
where we have used the notation
\begin{equation}
 \mathrm{T}_{n}^{\pm 1}\mathbf{h}
   = (h_1,\ldots, h_n \pm 1,\ldots, h_N) \quad \text{for } n\in\{1,\ldots,N\}.
\end{equation}
The action of $\mathcal{\bar{D}}_{-}(\lambda )$ on $\bra{\mathbf{h}_-}$ can then be obtained from the identity
\begin{equation}\label{id-DA}
\mathcal{\bar{D}}_{-}(\lambda )=\frac{2\lambda -\eta }{2\lambda }\mathcal{%
\bar{A}}_{-}(-\lambda )+\frac{\eta }{2\lambda }\mathcal{\bar{A}}_{-}(\lambda
).
\end{equation}
Similarly, the  action of $\mathcal{\bar{D}}_{-}(\lambda )$ on $\ket{\mathbf{h}_-}$ is given as
\begin{multline}
\mathcal{\bar{D}}_{-}(\lambda )\, \ket{\mathbf{h}_-}
     =\sum_{a=1}^{N}\sum_{\epsilon=\pm 1}
        \frac{ \big(2\lambda -\eta\big)\big(\lambda +\epsilon\xi_{a}^{(h_a)}\big) }{2\xi_{a}^{(h_a)} \big(2\xi_{a}^{(h_a)}-\epsilon\eta \big)}
        \prod_{\substack{ b=1  \\ b\neq a}}^{N}
        \frac{\lambda ^{2}-\big( \xi_{b}^{(h_b)} \big) ^{2}}{\big( \xi_{a}^{(h_a)}\big) ^{2}-\big( \xi_{b}^{(h_b)}\big)^{2}}
        \\
        \times
        k_a^\epsilon\,\mathsf{A}_{-}\big( -\epsilon\xi_{a}^{(1-h_a)} \big)\,
       \ket{\mathrm{T}_a^{\epsilon} \mathbf{h}_-}
       \\
       +\Bigg\{(-1)^N \mathrm{det}_{q} M(0)
       \prod_{b=1}^{N}\frac{\lambda ^{2}-\big( \xi_{b}^{(h_b)}\big) ^{2}}
                                            {\big( \eta /2\big) ^{2}-\big( \xi_{b}^{(h_b)}\big) ^{2}}
       -\frac{2\lambda-\eta}{2\bar{\zeta}_-}\prod_{b=1}^\mathsf{N}\Big(\lambda^2-\big(\xi_{b}^{(h_b)}\big)^2\Big)\Bigg\}
        \ket{\mathbf{h}_-},
      \label{act-D}
\end{multline}
and the action of $\mathcal{\bar{A}}_{-}(\lambda )$ on $\ket{\mathbf{h}_-}$ follows from
from the identity
\begin{equation}
\mathcal{\bar{A}}_{-}(\lambda )
=\frac{(2\lambda -\eta )}{2\lambda }\mathcal{\bar{D}}_{-}(-\lambda )
+\frac{\eta }{2\lambda }\mathcal{\bar{D}}_{-}(\lambda).  \label{A-decomp-right}
\end{equation}

It follows from the simplicity of the $\mathcal{\bar{B}}_{-}$-spectrum and from the above action of the operator $\bar{\mathcal{A}}_-(\lambda)$ (or $\bar{\mathcal{D}}_-(\lambda)$) that the left basis \eqref{left-SOV-state} and the right basis \eqref{right-SOV-state} are orthogonal with respect to the canonical scalar product in the spin basis, with
\begin{equation}\label{norm1}
     \moy{\mathbf{h}'_- \mid\mathbf{h}_-}
     = \delta_{\mathbf{h},\mathbf{h}'}\,\frac{  N_{\boldsymbol{\xi},-} }{\widehat{V}\big( \xi_1^{(h_1)} ,\ldots,\xi_N^{(h_N)}\big)}\,  .
\end{equation}
The normalization coefficient $N_{\boldsymbol{\xi},-}$ is equal to
\begin{align} 
     N_{\boldsymbol{\xi},-}
     &=\widehat{V}\big(\xi_1^{(0)} ,\ldots,\xi_N^{(0)}\big)\,  
     \bra{0}\prod_{n=1}^{N}    \frac{\mathcal{\bar{A}}_{-}(\eta /2-\xi _{n})}{\mathsf{A}_{-}(\eta /2-\xi _{n})}\,
     \ket{\underline{0}}
     \\
     &=\widehat{V}( \xi_1,\ldots, \xi_N)\,
    \frac{\widehat{V}( \xi_1^{(0)},\ldots, \xi_N^{(0)})}{\widehat{V}( \xi_1^{(1)},\ldots, \xi_N^{(1)})}
    \,\prod_{n=1}^N\frac{\bar{b}_{-}(\frac{\eta}{2}-\xi_n)}{\bar{a}_-(\frac{\eta}{2}-\xi_n)}.
\end{align}


\section{Proof of Identity~\ref{prop-id-3}}
\label{app-id}

This appendix is devoted to the proof of Identity~\ref{prop-id-3}, which follows
the same line of arguments as in the proof of Identity 3 of \cite{KitMNT16}.

Let us consider two sets of pairwise distinct variables $\{x\}\equiv\{x_1,\dots , x_M\}$ and $\{y\}\equiv\{y_1,\dots, y_L\}$ with $S=L-M\ge 0$.
For each $\epsilon\in\{+,-\}$, we introduce polynomials in $\lambda^2$ associated with the set of $M$ variables $x_1+\epsilon\frac{\eta}{2},\dots , x_M+\epsilon\frac{\eta}{2}$ as follows:
\begin{equation*}
     X^{(\epsilon)}(\lambda)=\prod_{\ell=1}^M\Big[\la^2-\Big(x_\ell+\epsilon\frac{\eta}{2}\Big)^{\!2\,}\Big],
     \quad \quad
     X^{(\epsilon)}_k(\lambda)=\frac{X^{(\epsilon)}(\lambda)}{\lambda^2-(x_k+\epsilon\frac{\eta}{2})^2},
     \quad 1\le k\le M.
\end{equation*}
We also introduce similar polynomials associated with a set of $S$ arbitrary pairwise distinct variables $w_1,\dots,w_S$:
\begin{equation*}
    W(\la)=\prod_{\ell=1}^S\left(\lambda^2-w_\ell^2\right), 
    \quad \text{and} \quad
     W_k(\lambda)=\prod_{\substack{\ell=1 \\ \ell\not= k}}^S\left(\lambda^2-w_\ell^2\right),
     \quad 1\le k\le S.
\end{equation*}
Using these polynomials, we define an auxiliary $(M+L)\times (M+L)$ matrix $\widetilde{\mathcal{C}}$ with coefficients $\widetilde{\mathcal{C}}_{j,k}$ given by the following relations:
\begin{alignat}{2}
   &X^{(+)}(\lambda)\ X_k^{(-)}(\lambda)=\sum_{j=1}^{M+L} \widetilde{\mathcal{C}}_{j,k}\, \lambda^{2(j-1)},
    & \qquad &1\le k\le \! M,  \nonumber\\
   &X_k^{(+)}(\lambda)\ X^{(-)}(\lambda)=\sum_{j=1}^{M+L} \widetilde{\mathcal{C}}_{j, M+k}\, \lambda^{2(j-1)},    
     & \qquad &1\le k\le \! M, \nonumber\\
   &X^{(+)}(\lambda)\ X^{(-)}(\lambda)\ W_k(\lambda)=\sum_{j=1}^{M+L} \widetilde{\mathcal{C}}_{j,2M+k}\, \lambda^{2(j-1)},
   & \qquad &1\le k\le S.
\end{alignat}
The determinant of this matrix can easily be computed. Indeed, we have
\begin{align*}
 &\sum_{j=1}^{M+L}\widetilde{\mathcal{C}}_{j,k}\, \Big(x_i-\frac{\eta}{2}\Big)^{2(j-1)}
 =\delta_{i,k} \ X^{(+)}\Big(x_i-\frac{\eta}{2}\Big)\ X_i^{(-)}\Big(x_i-\frac{\eta}{2}\Big),
 \\
 &\sum_{j=1}^{M+L}\widetilde{\mathcal{C}}_{j,k}\, \Big(x_i+\frac{\eta}{2}\Big)^{2(j-1)}
 =\delta_{i+M,k}\ X_i^{(+)}\Big(x_i+\frac{\eta}{2}\Big)\ X^{(-)}\Big(x_i+\frac{\eta}{2}\Big),
 \\
 &\sum_{j=1}^{M+L}\widetilde{\mathcal{C}}_{j,k}\, w_i^{2(j-1)}
 = \delta_{i+2M,k}\ X^{(+)}(w_i)\ X^{(-)}(w_i)\ W_i(w_i)+\widetilde{\widetilde{\mathcal{C}\,}}_{\! i,k},
\end{align*}
with $\widetilde{\widetilde{\mathcal{C}\,}}_{\! i,k}=0$ if $k>2M$, so that
\begin{multline}
 \widehat{V}\Big(x_1-\frac{\eta}{2},\ldots,x_M-\frac{\eta}{2},x_1+\frac{\eta}{2},\ldots,x_M+\frac{\eta}{2},w_1,\ldots,w_S\Big)\times \det_{M+L} \widetilde{\mathcal{C}}\\
 =\prod_{i=1}^M
  \left[ X^{(+)}\Big(x_i-\frac{\eta}{2}\Big)\,X_i^{(-)}\Big(x_i-\frac{\eta}{2}\Big)\,
          X_i^{(+)}\Big(x_i+\frac{\eta}{2}\Big)\,X^{(-)}\Big(x_i+\frac{\eta}{2}\Big) \right]
          \\
          \times
    \prod_{i=1}^S \left[ X^{(+)}(w_i)\, X^{(-)}(w_i)\,W_i(w_i) \right] ,
\end{multline}
which leads to
\begin{equation}
\det_{M+L}  \widetilde{ \mathcal{C}}
= \widehat{V}(w_S,\ldots,w_1)\, 
   \widehat{V}\Big(x_M+\frac{\eta}{2},\ldots,x_1+\frac{\eta}{2},x_M-\frac{\eta}{2},\ldots,x_1-\frac{\eta}{2}\Big). 
\end{equation}

Let us now  compute the product of $\mathcal{A}_{ \{x\}\cup\{y\} }[f]$ with the determinant of the matrix $\widetilde{ \mathcal{C}}$. It is given in terms of the determinant of a $(M+L)\times (M+L)$ matrix $\mathcal{G}$,
 \begin{equation}
    \mathcal{A}_{ \{x\}\cup\{y\} }[f] \cdot  \det_{M+L} \widetilde{\mathcal{C}}
  =\frac{ \det_{M+L}\mathcal{G} }{ \widehat{V}(x_1,\ldots,x_M,y_1,\ldots,y_L) },
 \end{equation}
 which can be expressed as the following block matrix:
 \begin{equation}
 \mathcal{G}
 =\begin{pmatrix}
 \mathcal{G}^{(1,1)}&\mathcal{G}^{(1,2)}&\mathcal{G}^{(1,3)}\\
 \mathcal{G}^{(2,1)}&\mathcal{G}^{(2,2)}&\mathcal{G}^{(2,3)}
 \end{pmatrix}.
 \end{equation}
 In this expression, the blocks $ \mathcal{G}^{(1,1)}$ and $\mathcal{G}^{(1,2)}$ are a $M\times M$ diagonal matrices:
\begin{align}
  \mathcal{G}^{(1,1)}_{i,k}
  &= \!\sum_{j=1}^{M+L}\widetilde{\mathcal{C}}_{j,k}
      \!\!\sum_{\epsilon\in\{+,-\}} \!\!
      f(\epsilon x_i)\,
     \left(x_i+\epsilon\frac{\eta}{2}\right)^{\! 2(j-1)}
     \nonumber\\
  &= \delta_{i,k}\  f(- x_i)\     X^{(+)}\Big(x_i-\frac{\eta}{2}\Big)\ X_i^{(-)}\Big(x_i-\frac{\eta}{2}\Big),
\end{align}
and
\begin{align}
  \mathcal{G}^{(1,2)}_{i,k}
  &= \!\sum_{j=1}^{M+L}\widetilde{\mathcal{C}}_{j,k+M}
      \!\!\sum_{\epsilon\in\{+,-\}} \!\!
      f(\epsilon x_i)\,
     \left(x_i+\epsilon\frac{\eta}{2}\right)^{\! 2(j-1)}
     \nonumber\\
  &= \delta_{i,k}\  f( x_i)\
     X_i^{(+)}\Big(x_i+\frac{\eta}{2}\Big)\ X^{(-)}\Big(x_i+\frac{\eta}{2}\Big),
\end{align}
for $1\le i,k\le M$.
The $M\times S$ block $\mathcal{G}^{(1,3)}$ vanishes:
\begin{align}
  \mathcal{G}^{(1,3)}_{i,k}
  &= \!\sum_{j=1}^{M+L}\widetilde{\mathcal{C}}_{j,k+2M}
     \!\!\sum_{\epsilon\in\{+,-\}} \!\!
      f(\epsilon x_i)\,     \left(x_i+\epsilon\frac{\eta}{2}\right)^{\! 2(j-1)}
     \nonumber\\
  &= 0,
  \qquad\qquad \forall \, i\in\{1,\ldots,M\},\ \forall\, k\in\{1,\ldots,S\}.
\end{align}
The $L\times M$ blocks $\mathcal{G}^{(2,1)}$ and $\mathcal{G}^{(2,2)}$ have the respective following forms:
\begin{align}
   \mathcal{G}^{(2,1)}_{i,k}
     &=\!\sum_{j=1}^{M+L}\widetilde{\mathcal{C}}_{j,k}
     \!\!\sum_{\epsilon\in\{+,-\}} \!\!
      f(\epsilon y_i)\,     \left(y_i+\epsilon\frac{\eta}{2}\right)^{\! 2(j-1)}
     \nonumber\\
    &=
     \!\sum_{\epsilon\in\{+,-\}}\!\!
      f(\epsilon y_i)\
     \frac{X^{(+)}(y_i+\epsilon\frac{\eta}{2})\ X^{(-)}(y_i+\epsilon\frac{\eta}{2})}
             {(y_i+\epsilon\frac{\eta}{2})^2-(x_k-\frac{\eta}{2})^2},
\end{align}
and
\begin{align}
   \mathcal{G}^{(2,2)}_{i,k}
     &=\!\sum_{j=1}^{M+L}\widetilde{\mathcal{C}}_{j,M+k}
     \!\!\sum_{\epsilon\in\{+,-\}} \!\!
      f(\epsilon y_i)\,     \left(y_i+\epsilon\frac{\eta}{2}\right)^{\! 2(j-1)}
     \nonumber\\
    &=
     \!\sum_{\epsilon\in\{+,-\}}\!\!
      f(\epsilon y_i)\
     \frac{X^{(+)}(y_i+\epsilon\frac{\eta}{2})\ X^{(-)}(y_i+\epsilon\frac{\eta}{2})}
             {(y_i+\epsilon\frac{\eta}{2})^2-(x_k+\frac{\eta}{2})^2},
\end{align}
for $1\le i\le L$ and $1\le k\le M$.
Finally, the $L\times S$ block  $\mathcal{G}^{(2,3)}$ is given by
\begin{align}
     \mathcal{G}^{(2,3)}_{i,k}
     &=\!\sum_{j=1}^{M+L}\widetilde{\mathcal{C}}_{j,2M+k}
     \!\!\sum_{\epsilon\in\{+,-\}} \!\!
      f(\epsilon y_i)\,     \left(y_i+\epsilon\frac{\eta}{2}\right)^{\! 2(j-1)}
     \nonumber\\
    &=
     \!\sum_{\epsilon\in\{+,-\}}\!\!
      f(\epsilon y_i)\      X^{(+)}\Big(y_i+\epsilon\frac{\eta}{2}\Big)\ X^{(-)}\Big(y_i+\epsilon\frac{\eta}{2}\Big)\
     W_k\Big(y_i+\epsilon\frac{\eta}{2}\Big)
     \nonumber\\
    &=\sum_{j=1}^S \mathcal{C}_{j,k}^W   
     \!\!\sum_{\epsilon\in\{+,-\}} \!\!\!
     f(\epsilon y_i)\,     X^{(+)}\Big(y_i+\epsilon\frac{\eta}{2}\Big)\, X^{(-)}\Big(y_i+\epsilon\frac{\eta}{2}\Big)
      \left(y_i+\epsilon\frac{\eta}{2}\right)^{\! 2(j-1)},
\end{align}
where $\mathcal{C}^W$ is the $S\times S$ matrix with determinant $\widehat{V}(w_S,\ldots,w_1)$ defined from the set of variables $\{w_1,\ldots,w_S\}$ by the relations
\begin{equation}
W_k(\lambda)=\sul_{j=1}^S \mathcal{C}^W_{j,k}\, \lambda^{j-1}.
\end{equation}
Hence
\begin{equation}\label{det-G-1}
    \det_{M+L}\mathcal{G}
    =\det_{M+L}\begin{pmatrix}
 \mathcal{G}^{(1,1)}&\mathcal{G}^{(1,2)}& 0\\
 \mathcal{G}^{(2,1)}&\mathcal{G}^{(2,2)}&\widetilde{\mathcal{G}}^{(2,3)}
       \end{pmatrix}
       \cdot \widehat{V}(w_S,\ldots,w_1),
\end{equation}
where the elements of the $L\times S$ block $\widetilde{\mathcal{G}}^{(2,3)}$ are
\begin{equation}
 \widetilde{ \mathcal{G}}^{(2,3)}_{i,k}=\!\!\sum_{\epsilon\in\{+,-\}} \!\!
      f(\epsilon y_i)\,
        X^{(+)}\Big(y_i+\epsilon\frac{\eta}{2}\Big)\, X^{(-)}\Big(y_i+\epsilon\frac{\eta}{2}\Big)
     \left(y_i+\epsilon\frac{\eta}{2}\right)^{\! 2(k-1)}.
\end{equation}
Finally, the remaining determinant in \eqref{det-G-1} can be computed by blocks using the fact that $\mathcal{G}^{(1,1)}$ is an invertible matrix and we obtain
\begin{equation}
  \det_{M+L}\begin{pmatrix}
 \mathcal{G}^{(1,1)}&\mathcal{G}^{(1,2)}& 0\\
 \mathcal{G}^{(2,1)}&\mathcal{G}^{(2,2)}&\widetilde{\mathcal{G}}^{(2,3)}
       \end{pmatrix}
  =\det_M \mathcal{G}^{(1,1)}\cdot \det_L\widetilde{\mathcal{G}},
\end{equation}
where
\begin{multline}
   \big[\, \widetilde{\mathcal{G}}\, \big]_{i,k}
   =\!\!\sum_{\epsilon\in\{+,-\}} \!\!
      f(\epsilon y_i)\ 
     X(y_i)\, X(y_i+\epsilon\eta)
     \\
     \times
     \begin{cases}
      {\displaystyle
      \, \left[ \frac{1}{(y_i+\epsilon\frac{\eta}{2})^2-(x_k+\frac{\eta}{2})^2}
     -\frac{\big(\mathcal{G}^{(1,1)}_{k,k}\big)^{-1}\, \mathcal{G}^{(1,2)}_{k,k}}{(y_i+\epsilon\frac{\eta}{2})^2-(x_k-\frac{\eta}{2})^2}\right]
     }
     \quad &\text{if } k\le M,
     \vspace{2mm}\\
    \, \left(y_i+\epsilon\frac{\eta}{2}\right)^{\! 2(k-M-1)}
     \qquad &\text{if } k >M.
     \end{cases}
\end{multline}
Here we have used notably the fact that $X^{(+)}(y_i+\epsilon\frac{\eta}{2})\ X^{(-)}(y_i+\epsilon\frac{\eta}{2})=X(y_i)\, X(y_i+\epsilon\eta)$ and that $\mathcal{G}^{(1,1)}$ and $\mathcal{G}^{(1,2)}$ are diagonal matrices.
Using their explicit form, we conclude the proof of Identity~\ref{prop-id-3}.



\end{document}